\documentclass[sigconf]{acmart} 
\UseRawInputEncoding
\usepackage{amsthm}
\usepackage{amsfonts} 
\usepackage{url}          
\usepackage{subfigure}
\usepackage{ifthen}
\usepackage{color, xcolor}
\hypersetup{colorlinks=true, linkcolor=blue, citecolor=blue, anchorcolor=blue, urlcolor=blue}
\usepackage{thm-restate}
\usepackage{algorithm}
\usepackage{algorithmic}
\usepackage{enumitem}

\newtheorem{theorem}{Theorem}
\newtheorem{definition}{Definition}
\newtheorem{corollary}{Corollary}
\DeclareMathOperator*{\argmax}{argmax}

\renewcommand\footnotetextcopyrightpermission[1]{}  
\newcommand{\E}{\mathbb{E}}
\newcommand{\I}{\mathbb{I}}
\newcommand{\R}{\mathbb{R}}

\newcommand{\bz}{\boldsymbol{z}}

\newcommand{\cI}{\mathcal{I}}

\newcommand{\cQ}{\mathcal{Q}}
\newcommand{\cR}{\mathcal{R}}
\newcommand{\cS}{\mathcal{S}}

\newcommand{\OPT}{{\it OPT}}
\newcommand{\EPT}{{\it EPT}}

\AtBeginDocument{%
  \providecommand\BibTeX{{%
    \normalfont B\kern-0.5em{\scshape i\kern-0.25em b}\kern-0.8em\TeX}}}

\setcopyright{acmcopyright}
\copyrightyear{2022}
\acmYear{2022}
\acmDOI{XXXXXXX.XXXXXXX}

\acmConference[SIGMOD '23]{the 2023 ACM SIGMOD/PODS Conference}{June 18--23, 2023}{Seattle, WA, USA}
\acmPrice{15.00}
\acmISBN{978-1-4503-XXXX-X/18/06}

\newcommand{\compilehidecomments}{false}
\ifthenelse{ \equal{\compilehidecomments}{true} }{%
	\newcommand{\wei}[1]{}
	\newcommand{\sheng}[1]{}
	\newcommand{\hao}[1]{}
	\newcommand{\zhou}[1]{}
        \newcommand{\zhang}[1]{}
}{
	\newcommand{\wei}[1]{{\color{blue!50!black}  [\text{Wei:} #1]}}
	\newcommand{\sheng}[1]{{\color{red!70!black} [\text{Sheng:} #1]}}
	\newcommand{\hao}[1]{{\color{green!90!black} [\text{Hao:} #1]}}
	\newcommand{\zhou}[1]{{\color{yellow!90!black} [\text{zhou:} #1]}}
        \newcommand{\zhang}[1]{{\color{purple!90!black} [\text{zhang:} #1]}}
}

\begin{document}

\title{Popularity Ratio Maximization: Surpassing Competitors through Influence Propagation}

\author{Hao Liao}
\orcid{}
\affiliation{%
  \institution{Shenzhen University}
  \city{Shenzhen}
  \country{China}}
  \email{haoliao@szu.edu.cn}

\author{Sheng Bi}
\affiliation{%
  \institution{Shenzhen University}
  \city{Shenzhen}
  \country{China}}
\email{2060271075@email.szu.edu.cn}

\author{Jiao Wu}
\affiliation{%
  \institution{Shenzhen University}
  \city{Shenzhen}
  \country{China}}
\email{1800271041@email.szu.edu.cn}

\author{Wei Zhang}
\affiliation{%
  \institution{Shenzhen University}
  \city{Shenzhen}
  \country{China}}
\email{2210275010@email.szu.edu.cn}

\author{Mingyang Zhou}
\affiliation{%
  \institution{Shenzhen University}
  \city{Shenzhen}
  \country{China}}
\email{zmy@szu.edu.cn}

\author{Rui Mao}
\affiliation{%
  \institution{Shenzhen University}
  \city{Shenzhen}
  \country{China}}
\email{mao@szu.edu.cn}

\author{Wei Chen}
\authornote{Contact author of this paper}
\affiliation{%
  \institution{Microsoft Research}
  \city{Beijing}
  \country{China}}
\email{weic@microsoft.com}

\renewcommand{\shortauthors}{Hao Liao and Sheng Bi, et al.}


\begin{abstract}
In this paper, we present an algorithmic study on how to surpass competitors in popularity by strategic promotions in social networks.
We first propose a novel model, in which we integrate the Preferential Attachment (PA) model for popularity growth with 
	the Independent Cascade (IC) model for influence propagation in social networks called PA-IC model.
In PA-IC, a popular item and a novice item grab shares of popularity from the natural popularity growth via the PA model, while the novice item
	tries to gain extra popularity via influence cascade in a social network. 
The {\em popularity ratio} is defined as the ratio of the popularity measure between the novice item and the popular item.
We formulate {\em Popularity Ratio Maximization (PRM)} as the problem of selecting seeds in multiple rounds to maximize the popularity ratio in the end.
We analyze the popularity ratio and show that it is monotone but not submodular.
To provide an effective solution, we devise a surrogate objective function and show that empirically it is very close to the original objective function while
	theoretically, it is monotone and submodular.
We design two efficient algorithms, one for the overlapping influence and non-overlapping seeds (across rounds) setting and the other for the non-overlapping influence and
	overlapping seed setting, and further discuss how to deal with other models and problem variants.
Our empirical evaluation further demonstrates that the proposed PRM-IMM method consistently achieves the best popularity promotion compared to other methods. 
Our theoretical and empirical analyses shed light on the interplay between influence maximization and preferential attachment in
social networks.
\end{abstract}




\maketitle

\section{Introduction}
Influence maximization (IM) is the problem of finding a subset of nodes (seed nodes) in a social network that could maximize the spread of influence~\cite{kempe03}.
It is a well-studied problem with applications on viral marketing, information propagation monitoring and control, etc (cf. \cite{chen2013information,LiFWT18}).
While most studies treat influence maximization as a stand-alone problem for viral marketing, 
	in this paper, we want to explore the means of using influence maximization
	to boost the popularity of an item (an idea, a product, etc) and surpass a competitor in a natural growth environment.
	
Consider the following hypothetical motivating example.
Alice opened her new restaurant, Caf\'{e} Alice, at a local shopping mall. 
However, a popular restaurant Bob's Kitchen is in the same mall. 
Without promotion, Alice would definitely lose many customers to Bob's Kitchen.
Thus Alice wants to distribute free-dish coupons to selected people, hoping them to propagate the information about Caf\'{e} Alice and attract more people to her restaurant,
	boosting the popularity of her restaurant in a short period of time to catch up or even surpass Bob's Kitchen.
Similar situations arise in online social media, such as a new blogger or podcaster trying to increase popularity through
	promotion in social networks.
Different from pure viral marketing, popularity growth has a natural rich-get-richer effect independent of the social network, meaning that the popular item would naturally attract more customers or users. 
For example, customers typically check nearby restaurants' popularity ratings on mobile apps such as Yelp for restaurant selection, making more people visit the popular restaurant.
This effect certainly gives a hard time for the new restaurant to catch up. 
Thus, incentive promotion by the new restaurant needs to incorporate this factor
	when deciding on viral marketing strategies.

This paper integrates the natural popularity growth with the incentive promotion into one coherent model to solve the above problem.
In particular, the preferential attachment (PA) \cite{yule1925ii,simon1955class,BA99} is a well-known model for the natural rich-get-richer effect of
	popularity growth, while the independent cascade (IC) model \cite{kempe03} is a classical model for influence propagation.
We integrate these two models into the novel PA-IC model to characterize the combined effect of natural popularity growth and incentive promotion.
More specifically, in a multi-round setting, we model the popularity growth of a novice item and a popular item.
We use $d_{t}^p$ and $d_{t}^n$ to denote the popularity of the popular and novice items respectively, at the end of round $t$. 
Each round, $z$ number of customers pick from these two items with probability proportional to their popularity, realizing the preferential attachment model.
Meanwhile, the novice item executes a promotion plan, selecting a set of seeds $S_t$ in round $t$ for $T$ consecutive rounds, and the influence spread generated
	by the seed set $S_t$ would add to the popularity measure of the novice item at the end of round $t$.
We define $r_t = d_{t}^n / d_{t}^p$ as the {\em popularity ratio} in round $t$.
The promotion task of the novice item is defined as the following {\em popularity ratio maximization (PRM)} problem: 
Given (a) the social network and its IC model parameters, (b) the initial
	popularity measures $d_{0}^p$ and $d_{0}^n$, and (c) a promotion budget $k$, find
	$k$ seeds and allocate them into $T$ consecutive rounds such that
	the popularity ratio $r_T$ at the end of round $T$ is maximized.
Depending on whether we allow overlapping seeds (OS) across rounds and whether repeated activation of the same nodes in different rounds are repeatedly counted (overlapping influence, OI), we further consider different settings such as the overlapping influence and non-overlapping seeds (OINS) and the non-overlapping influence and overlapping seeds (NIOS).

For both OINS and NIOS settings, we derive the formula for the popularity ratio and show that it is monotone but not submodular, 
	indicating that direct optimization on this objective function may be complex.
To provide an effective solution, we simplify the objective function and obtain a surrogate function and show that empirically it is very close to the original 
	objective function while theoretically, it is monotone and submodular.
Based on the reverse influence sampling approach ~\cite{BorgsBrautbarChayesLucier} and the Influence Maximization via Martingales(IMM) algorithm~\cite{tang15}, we design PRM-IMM algorithms 
	to solve this surrogate problem for the two settings, and prove their theoretical guarantees.
We further discuss how to handle other models and problem variants, such as when the popular item also conducts promotion, and when the novice item wants to use the minimum budget to
	surpass the popular item in a given time limit.
Our empirical evaluation demonstrates that the proposed PRM-IMM method always achieves the best popularity promotion compared to other methods. 

One important point we demonstrate both analytically and empirically is that the seed allocation of the novice item is tightly dependent on the natural growth of the environment as well as the popularity of the popular item.
This point indicates that the PRM task is quite different from the classical influence maximization --- the promotion of the novice item has to consider the popularity
	growth and its relative position against the popular item, while the classical influence maximization only targets 
	at increasing one's own influence spread or popularity.

In summary, our contributions include: 
	(a) proposing the PA-IC model that integrates preferential attachment with independent cascade propagation model;
	(b) formulating the PRM problem and studying its properties; and
	(c) designing an efficient algorithm to solve the PRM problem, and discussing the extension to other problem variants.
To the best of our knowledge, this is the first study that integrates preferential attachment and an influence propagation model into a coherent popularity growth model
 	and provides an algorithmic study on such an optimization problems.

\section{Related Works}

\paragraph{Preferential Attachment}
Yule firstly considered using the preferential attachment to explain the power-law distribution of flowering plants \cite{yule1925ii}.
A clearer and more general development of how preferential attachment leads to a power law was given by Simon \cite{simon1955class}. 
The Barabási-Albert (BA) model shows that the power-law degree distribution in the real networks can be produced by the combination of growth and preferential attachment \cite{BA99}.
Preferential attachment is now the standard model for the rich-get-richer effect.
%
%

{\em Influence Maximization.}
Domingos and Richardson first study influence maximization (IM) \cite{domingos01, richardson02}.
Then it is formulated as an optimization problem by Kempe et al. \cite{kempe03}, who
	define the famous independent cascade and linear threshold models, showing that influence maximization under these models are NP-hard, and
	connecting the problem to monotone and submodular function maximization.
Influence maximization has been extensively studied (see surveys~\cite{chen2013information,LiFWT18}).
One direction is to improve the scalability of the 
algorithms~\cite{Leskovec2007,ChenWY09,ChenWW10,ChenYZ10,wang2010community,JungHC12,simpath,kim2013scalable,BorgsBrautbarChayesLucier}.
The state-of-the-art scalable solution is based on the reverse influence sampling (RIS) approach first proposed by Borg et al.~\cite{BorgsBrautbarChayesLucier}, and
	later improved and refined by a series of studies~\cite{tang14,tang15,Nguyen_DSSA_2016,chen18,Tang_OPIM_2018}.
Our algorithms are also based on the RIS approach and is adapted from the IMM algorithm of~\cite{tang15}.

Another related direction is competitive influence maximization for multiple items~\cite{BAA11,HeSCJ12,goyal12stoc,chen2013foe,chen2013information,LCL2016}.
However, all these studies focus on the competitive diffusion of multiple items in social networks.
In contrast, the competition in our model stems from the popularity growth dictated by preferential attachment, and influence maximization is only a tool for
	the novice item to increase popularity.
To the best of our knowledge, we are the first to consider popular growth via both the PA effect and network diffusion, 
 and study their nontrivial interplay in an optimization problem.

%
%
%
%

\section{Model and Problem Definition}

In this section, we introduce the PA-IC model, which integrates the preferential attachment (PA) mechanism for popularity growth with the independent cascade (IC) model for influence propagation.
The PA mechanism models show two items, a popular item and a novice item, divide the shares of popularity from the natural growth of the customer base, while 
	the novice item further utilizes the IC model propagation to promote its popularity in order to catch up with the popular item.
Based on the PA-IC model, we define the optimization problem of {\em popularity ratio maximization}, which characterizes how to allocate 
	the promotional budget of the novice item to achieve the best promotion results.


\subsection{PA-IC Model}\label{sec:pa-ic model}
Preferential attachment (PA) \cite{BA99} is a well-known model to characterize the rich-get-richer effect in the growth of networks or popularity.
The basic form of the PA model is as follows.
Let $d^i_t$ be the popularity measure of item $i$ at time $t$.
When a number of $z$ users come at time $t$, each user selects item $i$ with probability proportional to $d^i_t$.
Thus popularity grows proportionally, i.e., $d^i_{t+1} = d^i_t + z \cdot d^i_t / \sum_j d^j_t$, where $z$ is  
the natural growth parameter.
%

Beyond the natural growth of popularity governed by the PA mechanism,
	we would like to incorporate the influence propagation model for popularity promotion.
Independent cascade (IC) model~\cite{kempe03} is a classical diffusion model widely adopted in the information and influence propagation literature, and its parameters can be 
	effectively learned in many applications (e.g. \cite{chen2013information,barbieri2012topic}).
In the IC model, a social network is modeled as a directed graph $G=(V,E,p)$, 
where $V$ is the set of $N = |V|$ nodes representing individuals,  
	$E$ is the set of $M = |E|$ directed edges representing influence relationships between pairs of individuals, 
and $p:E \rightarrow (0,1]$ gives the influence probability on every edge, i.e., $p(u,v)$ is the influence probability on edge $(u,v)\in E$.
Nodes have two states, active and inactive. Nodes that have been activated will always remain active. The propagation process starts from a seed set $S\subseteq V$ in discrete time steps. At step $\tau=0$, only the nodes in $S$ are active. 
At step $\tau \ge 1$, a node $u$ activated at step $\tau-1$ attempts to activate each of its inactive out-neighbor $v$ with success probability $p(u,v)$. 
Propagation terminates when no more nodes are activated. 
A key measure is the {\em influence spread} of the seed set $S$, denoted as $\sigma(S)$, which is the expected number of activated nodes at the end of propagation starting from the seed set $S$.
The IC model has the following equivalent live-edge graph description.
A random {\em live-edge graph} $L$ is sampled from $G=(V,E,p)$ such that every edge $(u,v)\in E$ has an independent probability of $p(u,v)$ to be included in $L$.
Then given a seed set $S$, the set of nodes activated in a stochastic diffusion from $S$ is the set of nodes reachable from $S$ in the live-edge graph $L$, denoted as $\Gamma(S,L)$.
Therefore, we have $\sigma(S)=\E_{L}[|\Gamma(S,L)|]$.

We integrate the PA model with the IC model in the following PA-IC model to 
	characterize both the natural growth and the promotional growth of the popularity.
For simplicity, this paper focuses on two competing items: a popular item and a novice item, since
	all items grow proportionally in the PA model.

The popular item starts with a higher popularity measure, while the novice item starts with a lower one.
If only natural growth is available, the novice item would never catch up with the popular item, and the gap would only be widened due to the rich-get-richer effect of the PA model.
To catch up and surpass the popular item, the novice item needs to employ incentive promotion in the social network to increase its popularity.
The technical description of the PA-IC model to cover the above aspects is given below.
	
In the PA-IC model, we model $T$ rounds of promotion and influence propagation.
Let $d_{t}^p$ and $d_{t}^n$ denote the popularity measure of the popular and novice items respectively, at the end of round $t$.
In each round $t \in [T]$, some customers would naturally select between the novice and popular item. 
We use the natural growth parameter $z$ to denote the number of customers.
These $z$ customers select between novice and popular items according to the PA model.

Besides the natural growth, in each round $t$, the novice item would select a seed set $S_t$ for its promotion, and these seeds will propagate influence throughout the social network $G$ following the IC model.
Note that in round $t$, the IC model may take multiple {\em steps} to propagate the promotion until the propagation ends.
The expected number of influenced nodes $\sigma^{\square}(S_t)$ is also added to the popularity measure of the novice item, where
	the actual form of $\sigma^{\square}(S_t)$ depends on the setting we consider, and it will be explained shortly after Eq.\eqref{eq:popularityeq} below.

The natural growth mechanism and propagation mechanism respectively model two different growth processes of an item's popularity in reality. Therefore, the natural growth mechanism and the propagation mechanism together give the following inductive formulation of the popularity growth of the two items.

\begin{align} \label{eq:popularityeq}
    &d_{t}^p = d_{t-1}^p + \frac{z\cdot d_{t-1}^p}{d_{t-1}^p+d_{t-1}^n}, d_{t}^n = d_{t-1}^n + \sigma^{\square}(S_t) + \frac{z\cdot d_{t-1}^n}{d_{t-1}^p+d_{t-1}^n}.
\end{align}

In Eq.\eqref{eq:popularityeq}, the notation $\sigma^{\square}$ needs to be instantiated to a concrete quantity according to the actual setting used. 
We mainly consider two settings --- the overlapping influence (OI) and
	the non-overlapping influence (NI) settings. 
In the overlapping influence (OI) setting, the influence spread between different rounds are allowed to be overlapped when they are added to the popularity measure.
In this case, $\sigma^{\square}(S)$ is instantiated to $\sigma^{OI}(S)$, and is simply the standard influence spread $\sigma(S)$.
In the non-overlapping Influence (NI) setting, the influence spread of $S_t$ in round $t$ cannot overlap with the influence spread in the previous rounds, and thus
	it is the marginal influence spread of $S_t$ given the seed sets $S_1, \ldots, S_{t-1}$ of the previous rounds. 
Hence, in this case $\sigma^{\square}(S)$ is instantiated to $\sigma^{NI}(S_t|S_1,\ldots, S_{t-1})$, as defined below:
\begin{equation}\label{eq:influence spread(NI)}
    \sigma^{NI}(S_t|S_1\cdots S_{t-1}) = \mathbb{E}_{L_1,\ldots, L_t}\left[\left|\Gamma(S_t, L_t)\setminus \bigcup_{i=1}^{t-1}\Gamma(S_i,L_i)\right| \right].
\end{equation}
where $L_i$, $i\in [t]$ is the random live-edge graph in round $i$.
When the context is clear, we use $\sigma^{NI}(S_t)$ as a shorthand for $\sigma^{NI}(S_t|S_1\cdots S_{t-1})$.

One remark to Eq.\eqref{eq:popularityeq} is that both the natural growth and the social network promotion parts are represented in the expectation form.
One may formulate them as random variables, and take the expectation in the end.
The expectation form is easier to handle, and our empirical evaluation would show that such representation does not lose fidelity
	in terms of the solution quality.

\begin{figure*}[!t]
    \centering
    \includegraphics[scale=0.46]{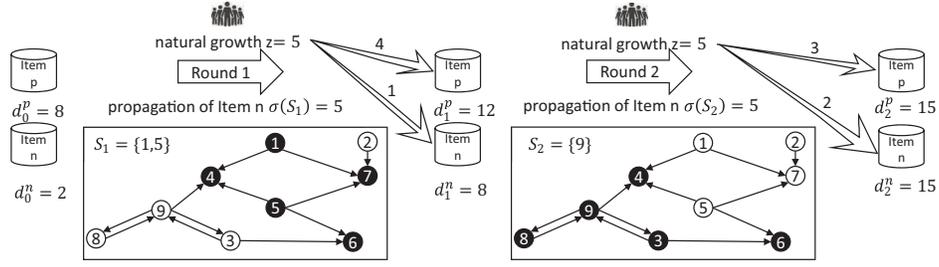}
    \caption{A numerical example of the PA-IC model with the overlapping influence setting. 
    	In this example, we show the changes in popularity measure of two items, within two rounds. In this directed graph, every edge's activation probability is equal to 1. Black nodes is the "active" nodes in each round.}
    \label{fig:example}
\end{figure*}

Another remark concerns the distinction of the natural growth and promotion in our model.
The natural growth follows the preferential attachment and is independent of the social network, while the
	promotion follows the influence diffusion model and is dependent on the social network.
One may say that the natural growth may also generate influence diffusion in social networks.
The reason we do not include the social network diffusion for the natural growth part of the model is that
	(a) the natural growth mainly reflects the rich-get-richer effect and models the behavior that people gets the
	popularity information from online review platforms such as Yelp or direct offline observations (e.g., by
	observing the occupancy of the restaurants) and act accordingly, and this part is not related to social network
	diffusion; and
	(b) inclusion of the influence diffusion for the natural growth part will further complicate the model, making it
	less clear in demonstrating the interaction between the preferential attachment for natural
	growth and the influence diffusion for viral promotion.

\begin{table*}[!t]
	\caption{Four settings of the PRM problem}
	\label{Four settings}
	\begin{tabular}{lcc}\\
		\toprule[2pt]
		&Overlapping Influence&Non-overlapping Influence\\
		\hline
		Overlapping Seeds& OIOS & NIOS \\
		Non-overlapping Seeds& OINS & NINS \\
		\bottomrule[2pt]
	\end{tabular}
\end{table*}

Fig. \ref{fig:example} is a numerical example to illustrate our PA-IC model in the OI setting. 
There are two items in this social network: item p is the popular item with the initial popularity measure $d_0^p = 8$, item n is the novice item with the initial popularity measure $d_0^n = 2$, and the natural growth $z = 5$. All edges in the social network have probability $1$.
In round $1$, according to the PA mechanism, the popularity measure of item p will increase by $4$, and the popularity measure of item n will increase by $1$. At the same time, the novice item chooses the seed set $S_1 = \{1,5\}$ to promote their product and the influence spread $\sigma(S_1) = 5$. So the popularity measure of item n will additionally increase by $5$. At the end of round $1$, $d_1^p = d_0^p + \frac{z\cdot d_0^p}{d_0^p+d_0^n} = 12$, $d_1^n = d_0^n + \frac{z\cdot d_0^n}{d_0^p+d_0^n} + \sigma(S_1)=8$. In round $2$, according to the PA mechanism, the popularity measure of item p will increase by $3$, and the popularity measure of item n will increase by $2$. In this round, item n chooses the seed set $S_2 = \{9\}$ to promote their product, which generates influence spread $\sigma(S_2)=5$. So at the end of round $2$, $d_2^p = d_2^n = 15$, which means that the novice item has caught up with the popular item.

Finally, the model can be further extended to allow different natural growth count $a_t$ for each round $t$, or allow the popular item to also have a viral promotion mechanism.
We defer the discussion of these extensions and their impacts on our algorithm to Section~\ref{sec:variants}.

\subsection{Popularity Ratio Maximization}\label{sec:problem definition}
We define the {\em popularity ratio} between the novice and the popular item at the end of round $t$ as $r_t = d^n_t / d^p_t$.
Without the promotional mechanism, this ratio would not change.
Thus, from the novice item's perspective, it wants to increase this ratio as much as possible and as soon as possible through promotion, but it certainly has budget constraints.

Technically, we model this as a popularity ratio maximization problem under a budget constraint.
We use the pair $(v,t)$ to denote that the novice item selects $v$ in round $t$ as a seed, i.e., $v \in S_t$, and
	we use $\cS = S_1\times \{1\} \cup S_2\times \{2\} \cdots  \cup S_T\times \{T\}$ to represent the overall allocation of seeds over $T$ rounds.
The novice item has a total budget $k$, restricting the total number of seeds that it can select, i.e., $|\cS| \le k$. 
Let the final popularity ratio after $T$ rounds promotion be $r_T(\cS) = d_T^n$/$d_T^p$.

Informally, popularity ratio maximization (PRM) is to find a seed allocation $\cS$ with $|\cS|\le k$ for the novice item 
	so that popularity ratio  $r_T(\cS)$ at the end of round $T$ is maximized.
In the seed allocation $\cS$, the seed set for different round may or may not be allowed to overlap, and both may be reasonable depending on the application scenario.
Thus, together with the overlapping or non-overlapping influence settings, there are four settings as shown in Table~\ref{Four settings}, where OIOS refers to overlapping influence and overlapping seeds and NINS refers to non-overlapping influence and non-overlapping seeds. The other two settings have been defined before.
Formally, we define the PRM problem below based on these settings.


\begin{definition}[Popularity Ratio Maximization] \label{def:PRM}
Given (a) a social network $G=(V,E,p)$, (b) initial popularity measures $d^n_0$ and $d^p_0$ for the novice and popular items respectively, 
	(c) total round number $T$, (d) natural growth count $z$, (e) budget $k$, 
	the task of {\em popularity ratio maximization (PRM)} is to find an optimal seed set allocation $\cS^*$ for the novice item in one of the four settings depicted in Table~\ref{Four settings}, 
	such that the total number of seeds does not exceed $k$, and when the popularity measures evolve according to Eq.\eqref{eq:popularityeq}, the final popularity ratio $r_T(\cS^*)$ is maximized. 
	For the OINS setting, it is
\begin{displaymath}
\cS^* \in \argmax_{\cS \subseteq V\times [T], |\cS|\leq k, S_t\cap S_{t'} = \emptyset \forall t\neq t'} r_T\left(\cS\right).
\end{displaymath}
Other settings can be similarly formulated.
\end{definition}

Note that it may not be wise to use all the budget in the first round since too many seeds could generate redundant activations of the same node, which is only counted once in the popularity measure.
On the other hand, spreading the budget evenly across all $T$ rounds may not be a good choice either, since it would make the popularity growth of the novice item slower and thus harder to catch up with the popular item.
Moreover, even though the objective is to increase the popularity of the novice item, the problem is still tightly related to the popularity of the popular item, because
	the PA mechanism links them together.
Therefore, PRM is non-trivial both in considering seed allocation across multiple rounds and in considering the relationship with the popular item.
Furthermore, PRM needs to consider the impact of different settings.
In the following two sections, we will consider two representative settings --- overlapping influence and non-overlapping seeds (OINS) and non-overlapping influence and overlapping seeds (NIOS).
Other settings can be treated similarly. 

PRM can be extended to other variants, and we will consider a number of them in Section~\ref{sec:variants}.

\section{Results on the OINS Setting}\label{sec:OINS}
We study the OINS setting in this section.
We first derive the exact formula for the objective function of the PRM problem.
We show that the exact objective function is not submodular, implying a rather difficult optimization task.
To tackle this problem, we provide a heuristic simplification of the objective function, 
	resulting in a round-weighted influence objective function $\rho_T(\cS)$, and its
	corresponding surrogate optimization task called {\em round-weighted influence maximization (RWIM)}.
We show that function $\rho_T(\cS)$ is monotone and submodular, which allows us to design
	an efficient approximation algorithm PRM-IMM based on the IMM algorithm~\cite{tang15}. 
The technical novelty of PRM-IMM includes properly defining the pair-wise reverse reachable set (PW-RR set) to 
	estimate the seed set's influence in the different rounds with different weights, 
	and calculating the number of PW-RR sets needed to satisfy the approximate ratio.
Although our approximation guarantee is on the surrogate RWIM algorithm, our experimental evaluation will demonstrate that 
	it also solves the original PRM problem with the best performance compared to other baselines.

\subsection{Objective Function under the OINS Setting}
\label{sec:OINSobj}
The following lemma states the exact formula for the objective function $r^{OI}_T(\cS)$ of PRM under the overlapping influence (OI) setting, which is derived recursively.
\begin{lemma} 
\label{lem:OINSobjective}
For the OINS setting, the popularity ratio function at the end of round $T$ is:
\begin{equation}\label{eq:ratio plus one(OINS)}
    r^{OI}_T(\cS)  = \left(r_0+1\right)\prod_{t=1}^T \left(1+\frac{\sigma(S_t)}{d_0^n+d_0^p+z\cdot t+\sum_{i=1}^{t-1} \sigma(S_i)}\right)-1,
\end{equation}
where $r_0 = d_0^n/d_0^p$ is the initial popularity ratio, $\cS = \bigcup_{t=1}^T S_t \times \{t\}
$.
\end{lemma}
\begin{proof}
We show Eq.\eqref{eq:ratio plus one(OINS)} by deriving the recursive formula for $r_t$, $t\ge 1$.
Suppose at the end of round $t-1$ the popularity ratio is $r_{t-1}$, and the popularity measures for the novice and popular
	items are $d_{t-1}^n$ and $d_{t-1}^p$, respectively. 
By definition, we have
\begin{align}
r_t + 1 & = \frac{d_t^n+d_t^p}{d_t^p}  = \frac{d_{t-1}^p+d_{t-1}^n + \sigma^{OI}(S_t) + z}{d_{t-1}^p + z\cdot\frac{d_{t-1}^p}{d_{t-1}^p+d_{t-1}^n}} \label{eq:recursivepopularity} \\
& = \frac{d_{t-1}^p+d_{t-1}^n + \sigma^{OI}(S_t) + z}{d_{t-1}^p+d_{t-1}^n + z} \cdot \frac{d_{t-1}^p+d_{t-1}^n}{d_{t-1}^p} \nonumber \\
& = \left(1+\frac{\sigma^{OI}(S_t)}{d_0^n+d_0^p+z\cdot t+\sum_{i=1}^{t-1} \sigma^{OI}(S_i)}\right) (r_{t-1}+1), \label{eq:sumpopularity}
\end{align}
where the second equality in \eqref{eq:recursivepopularity} is from Eq.\eqref{eq:popularityeq}, and
	Eq.\eqref{eq:sumpopularity} uses the fact that $d_{t-1}^p+d_{t-1}^n$ is the total popularity measure of the two items
	at the end of round $t-1$,
	which is accumulated from $d_0^n+d_0^p$ by $t-1$ rounds of natural growth ($t\cdot z$) and 
	the promotion effect ($\sum_{i=1}^{t-1} \sigma^{OI}(S_i)$).
Eq.~\eqref{eq:ratio plus one(OINS)} can be obtained immediately with the above recursive formula and the fact that
	$\sigma^{OI}(S_i) = \sigma(S_i)$.

\end{proof}

One feature of Eq.~\eqref{eq:ratio plus one(OINS)} is that the popularity ratio depends on 
	the interplay among the parameters of the novice item ($d_0^n$), the popular item ($d_0^p$), the natural growth ($z$) and 
	the total influence effect $\sigma(S_t)$, which makes the seed selection more complicated in the PRM task than the classical
	influence maximization task.

A set function $f: 2^V\rightarrow \R$ is monotone if $f(S)\le f(Q)$ whenever $S\subseteq Q \subseteq V$, and
	$f$ is submodular if, for any $S\subseteq Q \subseteq V$ and any $x \in V\backslash Q$, $f(S\cup \{x\})-f(S)\geq f(Q\cup \{x\}) - f(Q)$. 
Submodularity means that the marginal value of an element with respect to a set $S$ decreases as $S$ grows.
Maximizing a monotone and submodular set function with a cardinality constraint
	can be achieved by a simple greedy algorithm with an approximation ratio
	of $1-1/e$~\cite{NWF78}.
Unfortunately, the exact popularity ratio function in Eq.~\eqref{eq:ratio plus one(OINS)} is not submodular, as shown below.
\begin{restatable}{lemma}{lemnonsubmodular} \label{lemma:submodular of popularity ratio}
The popularity ratio function $r^{OI}_T(\cS)$ is monotone but not submodular. 
\end{restatable}
The proof is omitted, and the non-submodularity is due to the multiplication of submodular functions in Eq.\eqref{eq:ratio plus one(OINS)}.

\subsection{Round-Weighted Influence Maximization}
\label{sec:simplify}
Due to the non-submodularity of the popularity ratio function, it is challenging to design a good algorithm for the PRM problem directly.
In this section, we simplify the popularity ratio function to obtain a surrogate function and show that this surrogate function is monotone and submodular.
	
As indicated in the proof sketch of Lemma \ref{lemma:submodular of popularity ratio}, the complication of the original 
	objective function of Eq.\eqref{eq:ratio plus one(OINS)} is the series of multiplications, and thus this is the target of our simplification.
In particular, our simplification consists of two steps: 
	(a) expanding the multiplication series of Eq.\eqref{eq:ratio plus one(OINS)} and only keeping the first-order terms;
	(b) removing the $\sigma(S_1),\ldots, \sigma(S_{T-1})$ in the denominator of each term left after step (a). 
Step (a) decreases the objective value, while step (b) increases the objective value, compensating the effect of step (a).
Our experimental evaluation (Section~\ref{sec:modeljustify}) shows that the surrogate function we develop is reasonably close to 
the original function in all the test cases.
The following is the formula for the surrogate function (ignoring the constant factor $r_0+1$ and the additive terms in Eq.\eqref{eq:ratio plus one(OINS)}):
\begin{displaymath}
    \rho^{OI}_T(\cS) = \frac{\sigma(S_1)}{d_0^n+d_0^p+z} + \frac{\sigma(S_2)}{d_0^n+d_0^p+2z} + \cdots + \frac{\sigma(S_T)}{d_0^n+d_0^p+T\cdot z}.
\end{displaymath}
In the above formula, we can set
\begin{equation}
\label{eq:setweight}
 w_t = \frac{1}{d_0^n+d_0^p+t\cdot z}
\end{equation}
as the weight for $\sigma(S_t)$ for every round $t$, and turn the objective
	function into the following more general objective function:
\begin{equation}\label{eq:weighted im}
    \rho^{OI}_T(\cS)=\sum_{t=1}^{T} {w_{t}\cdot\sigma^{OI}(S_{t})} = \sum_{t=1}^{T} {w_{t}\cdot\sigma(S_{t})}.
\end{equation}
Function $\rho^{OI}_T(\cS)$ means that we need to find an allocation of seeds of $T$ rounds and each round $t \in [T]$ has its own importance
	weight $w_t$, which is non-negative and decreasing, i.e., $w_1 \ge w_2 \ge \cdots \ge w_T \ge 0$.
We call $\rho^{OI}_T(\cS)$ the {\em round-weighted influence function}.
	

The relationship between the popularity ratio function and round-weighted influence function is $r^{OI}_T(\cS)\approx(1+\rho^{OI}_T(\cS))(r_0+1)-1$. 
The surrogate function is summarized below as the round-weighted influence maximization problem.

\begin{definition}[Round-weighted Influence Maximization in the OINS Setting]
Given (a) a social network $G=(V,E,p)$, (b) initial popularity measures $d^n_0$ and $d^p_0$ for the novice and popular items respectively, 
(c) total round number $T$, (d) natural growth count $z$, (e) budget $k$, 
(f) the non-negative and non-increasing weight sequence $w_1, w_2, \ldots, w_T$,
the task of {\em round-weighted influence maximization (RWIM)} in the OINS setting
	is to find an optimal seed set allocation $\cS^*$ for the novice item
	under the constraints that the seeds selected for each round are disjoint, and
	the total number of seeds does not exceed $k$, such that 
	the round-weighted influence $\rho^{OI}_T(\cS^*)$ is maximized. That is,
\begin{displaymath}
\cS^* \in \argmax_{\cS \subseteq V\times [T], |\cS|\leq k, S_t\cap S_{t'} = \emptyset \forall t\neq t'} \rho^{OI}_T\left(\cS\right).
\end{displaymath}
\end{definition}

\begin{lemma} \label{lem:rhoproperty}
For any influence graph $G=(V, E, p)$, under the PA-IC model, the round-weighted influence function 
	$\rho^{OI}_T(\cS)$ defined in Eq.\eqref{eq:weighted im} is monotone and submodular.
\end{lemma}

\subsection{Pair-wise Reverse Reachable Sets}\label{sec:PW-RR}
State-of-the-art IM algorithms, including IMM, use the reverse influence sampling (RIS) approach governed by reverse reachable (RR) sets. 
An RR set is a random set of nodes sampled from the graph by (a) first selecting a node $v$ uniformly at random from the graph, and 
(b) simulating the reverse propagation of the model (e.g., IC model) and adding all visited nodes into the RR set. 
The main property of a random RR set R is that: influence spread $\sigma(S) = N \cdot \mathbb{E}\left[\mathbb{I}\{S \cap R\neq \emptyset\}\right]$ for 
any seed set $S$, where $\mathbb{I}$ is the indicator function, $N$ is the number of nodes in graph $G$. 
After finding a large enough number of RR sets, the original influence maximization problem is turned into a $k$-max coverage problem ---
finding a set of $k$ nodes that covers the most number of RR sets, where a set $S$ covers an RR set $R$ means that $S\cap R\neq \emptyset$. 
All RIS algorithms have the following two steps: 
1. Generate a sufficiently large set of random RR sets, called Sampling step.
2. Find $k$ nodes that cover the most number of RR sets, called NodeSelection step.

%
We now introduce our adaptation of the RR set to the PRM setting, particularly, the pair-wise reverse reachable (PW-RR) set.

\textbf{PW-RR set and its generation process.}
Since the RR set is not designed to estimate the influence of nodes in different rounds, we define the pair-wise reverse reachable (PW-RR) set for the PRM problem and denote it as $R^{(t)}$. 
PW-RR can distinguish between different rounds of RR set, and it also corresponds to different weights for different rounds of RR sets so that the round-weighted influence of nodes in different rounds can be calculated.
A (random) pair-wise RR set $R^{(t)}$ is an RR set $R$ rooted at a node picked uniformly at random from $V$, and $t$ is picked uniformly at random from $[T]$,
	and $R^{(t)} = R \times \{t\}$.

\textbf{Property of PW-RR set}
The main property of random PW-RR set $R^{(t)}$'s is that we can use them to estimate the round-weighted influence spread $\rho^{OI}_T(\cS)$.

\begin{lemma}\label{lemma:property of PW-RR}
	For any set of node-round pairs $\cS$, we have
	\begin{equation} \label{eq:rhoidentity}
	\rho^{OI}_T(\cS) = N\cdot T \cdot  \mathbb{E}\left[w_t\cdot \mathbb{I}\left[\cS \cap R^{(t)} \neq \emptyset \right] \right],
	\end{equation}
	where $R^{(t)}$ is a random PW-RR set, and the expectation is taken from the randomness of
		(a) the root of a PW-RR set $R^{(t)}$ uniformly chosen at random,
		(b) the randomly generated RR set from the root, and
		(c) round number $t$ uniformly chosen at random.
\end{lemma}
\begin{proof}
	The lemma is shown by the following:
	\begin{align}
	& \mathbb{E}\left[w_t\cdot \mathbb{I}\left[\cS \cap R^{(t)} \neq \emptyset \right] \right] 
	= \frac{1}{T} \sum_{t=1}^{T} w_t \cdot \mathbb{E}\left[ \I \{S_t \cap R \neq \emptyset \} \right]  \label{eq:enumeratet}\\
	&=\frac{1}{T} \sum_{t=1}^{T} w_t \cdot \frac{1}{N} \sigma\left(S_t\right) = \frac{1}{N \cdot T} \rho^{OI}_T(\cS), \label{eq:oldidentity}
	\end{align}
	where the $R$ in Eq.\eqref{eq:enumeratet} is a random RR set as defined at the beginning of this section, the first equality
	in \eqref{eq:oldidentity} is by the property of the RR set as given at the beginning of this section, and the second equality in \eqref{eq:oldidentity} is
	by Eq.\eqref{eq:weighted im}.

\end{proof}
Lemma~\ref{lemma:property of PW-RR} reflects the overlapping influence (OI), and it is clearly seen in Eq.\eqref{eq:oldidentity} where influence spread of different rounds are summed up,
	allowing overlapping influence.

\subsection{Popularity Ratio Maximization IMM} \label{sec:PRMIMM}
We could base the key formula of Eq.\eqref{eq:rhoidentity} to estimate the overlapping influence spread $\rho^{OI}_T(\cS)$.
Let $\cR$ be a collection of $\theta$ PW-RR sets we generated independently, $\cR = \{R^{(t_1)}_1, R^{(t_2)}_2, \ldots, R^{(t_\theta)}_\theta\}$.
Define random variable $Y^{\cR}_i(\cS) = w_{t_i}\cdot \mathbb{I}\left[\cS \cap R^{(t_i)}_i \neq \emptyset \right]$.
We can estimate $\rho^{OI}_T(\cS)$ using the following estimator $\hat{\rho}^{OI}_T(\cS,\cR)$:
\begin{displaymath}
    \hat{\rho}^{OI}_T\left(\cS, \cR\right)=\frac{N \cdot T}{\theta} \sum_{i=1}^{\theta} Y_i^{\cR}(\cS).
\end{displaymath}

Estimator $\hat{\rho}^{OI}_T(\cS,\cR)$ can be viewed as a weighted coverage function
on a bipartite graph with $V\times [T]$ on the one side and PW-RR sets $\cR$ on the other side.

Our PRM-IMM algorithm follows the IMM structure~\cite{tang15} and consists of two components:
PRM-NodeSelection for selecting top $k$ pairs from the above bipartite graph, and PRM-Sampling for generating enough PW-RR set samples.
We now explain the main adaption of these two components to our PRM settings.



\begin{algorithm}[t]
\caption{PRM-NodeSelection-OINS}
\label{alg:nodeselection}
\leftline{\textbf{Input:} PW-RR sets $\cR$, budget $k$, weights $(w_t)_{t\in[T]}$ (Eq.\eqref{eq:setweight})}
\leftline{\textbf{Output:} seed set $\hat{\cS}^g$}
\begin{algorithmic}[1]
\STATE $RR[(v,t)] = \{R\in \cR | (v,t) \in R\}$, $c[(v,t)]= w_t \cdot |RR[(v,t)]|$, $\forall (v,t)\in V\times [T]$ 
	/* $RR[(v,t)]$'s and $c[(v,t)]$'s can be constructed during the generation of $\cR$ */
\STATE for all $R\in\cR, t \in [T], covered[R]= {\it false}$
\STATE $\cS=\emptyset$ /* the element of $\cS$ is $(v,t)$ */
\FOR{$i=1$ to $k$}        
    \STATE \label{line:argmax1} 
    	$(v, t)=\operatorname{argmax}_{(v, t) \in\left(V \backslash \cS_{\text {node}}\right) \times[k]}
    c[(v,t)]$      \label{line 1}
    \STATE /*$\cS_{node}$ is the node set of $\cS$ */   
    \STATE $\cS = \cS \cup \{(v,t)\}$ 
    \FOR{$R \in RR[(v,t)] \wedge covered[R]== {\it false} $} 
        \STATE $covered[R] = true$
        \FOR{$(u,t) \in R \wedge u \neq v $}
            \STATE $c[(u,t)] = c[(u,t)] - w_t$
        \ENDFOR
    \ENDFOR
\ENDFOR
\RETURN $\hat{\cS}^g = \cS $
\end{algorithmic}
\end{algorithm}

\textbf{PRM-NodeSelection-OINS}
Algorithm ~\ref{alg:nodeselection} provides the pseudocode for the PRM-NodeSelection under the OINS setting,
which is patterned from the similar algorithms in~\cite{tang15,chen20}.
It mainly implements a greedy selection of $k$ pairs to cover the PW-RR sets.
Several points worth mentioning are:
(a) the greedy algorithm works on PW-RR sets for $(v,t)$ pairs, as described in Section~\ref{sec:PW-RR};
(b) the overlapping influence (OI) aspect is reflected in the PW-RR set design and Lemma~\ref{lemma:property of PW-RR}; and
(c) the non-overlapping seed (NS) aspect is reflected in line~\ref{line:argmax1}, where nodes already selected, $\cS_{\text {node}}$, are removed from
further consideration in seed selection.

Due to the NS requirement, the feasible solutions for the bipartite coverage problem actually form a partition matroid.\footnote{A matroid $(E,\cI)$ on a finite ground element set $E$ is a collection $\cI$ of subsets of $E$,
	each called an independent set, with the following two properties: (a) a subset of an independent set is still an independent set; and (b) for two independent sets
	$I_1,I_2 \in \cI$, if $|I_1| < |I_2|$, then there is at least one element $e\in I_2 \setminus I_1$ such that $I_1 \cup \{e\}$ is also an independent set.
	A partition matroid is a matroid where the ground set $E$ is partitioned into disjoint sets $E_1, E_2, \ldots, E_s$, and each $E_i$ has a capacity constraint $b_i$ such that 
	for all $I \in \cI$, $|I \cap E_i| \le d_i$.
}
The greedy algorithm on the partition matroid gives a $1/2$ approximation of the optimal solution.\footnote{We are aware that there is a polynomial-time algorithm that achieves $(1-1/e)$-approximation for submodular maximization under a matroid constraint~\cite{CalinescuCPV11}, but the algorithm is highly inefficient, and thus we choose the simple greedy algorithm for our implementation.
}

\textbf{PRM-Sampling}
The sampling procedure also follows the IMM structure (Algorithm \ref{alg:PRM-IMM}), in that it uses iterative halving to estimate a lower bound $LB$ of the optimal solution $\OPT$, and then use
	$LB$ to determine the final number of PW-RR sets needed.
The parameters $\alpha,\beta$ are defined below:

\begin{equation} \label{eq:alphabeta}
	\alpha=\sqrt{l \ln N+\ln 4}~,\quad \beta=\sqrt{\frac{1}{2} \cdot\left(\ln \binom{N}{k} + \alpha^2 +k \ln T\right)}.
\end{equation}

There are a few differences compared to the classical IMM.
First, the $\theta_i$ (line~\ref{line:theta}) and $\tilde{\theta}$ (line~\ref{line:tildetheta}) have additional coefficients $w_1$ and $T$ in the numerator. 
Factor $T$ is because we allocate seeds in $T$ rounds.
Factor $w_1$ is technically because random variable $Y^{\cR}_i(\cS) / w_1$ is bounded within $[0,1]$, so that we can apply
	the Chernoff bound on it to obtain the desired result.
Second, the parameter $\beta$ also adds an extra $k\ln T$, because in our problem, the number of possible solutions is upper bounded
	by $\binom{N}{k}\cdot T^k$, which means that $k$ nodes are selected among $N$ nodes
	and each selected node may be placed in any of the $T$ rounds.
Of course, our approximation guarantee is also changed from $1-1/e-\varepsilon$ to $1/2 -\varepsilon$, because we are dealing with a more general partition matroid constraint.
The regeneration of PW-RR sets in line~\ref{line:regenerateRRsets} ensures that all PW-RR sets in $\cR$ are mutually independent \cite{chen18}.

\begin{algorithm}[t]
\caption{PRM-IMM-OINS algorithm}
\label{alg:PRM-IMM}
\leftline{\textbf{Input:}
directed graph $G=(V,E)$, round number $T$, budget $k$,} 
\leftline{weights $(w_t)_{t\in[T]}$ (Eq.\eqref{eq:setweight}), approximate ratio parameter $\varepsilon$,}
\leftline{Error probability parameter $l$.}
\leftline{\textbf{Output:}  $\hat{\cS}^{g}$}
\begin{algorithmic}[1]
\STATE $\cR=\emptyset, LB=w_1, \varepsilon' =\sqrt{2}, \theta_0 = 0$  
\FOR{$i=1$ to $\left\lfloor\log_{2}N \right\rfloor -1 $}
\STATE $x_i = \frac{\sum_{t=1}^{T} w_t \cdot N}{2^i}$ \label{alg:PRM-IMM 1}

\STATE \label{line:theta}
	$\theta_i = \left\lceil\frac{w_1\cdot N\cdot T\left(2+\frac{2}{3} \varepsilon^{\prime}\right)\left(\ln\log_{2}N+2\beta^2\right)}{\varepsilon^{\prime^{2}} x_{i}}\right\rceil$
	
\STATE independently generate $\theta_i - \theta_{i-1}$ PW-RR sets into $\cR$: 
\STATE $\cS$ = PRM-NodeSelection-OINS$(\cR,k, (w_t)_{t\in[T]})$
\IF{$\hat{\rho}_T\left(\cS,\cR\right) \geq \left(1+\varepsilon'\right)x_i$} 
\STATE $LB = \frac{\hat{\rho}_T\left(\cS,\cR\right)}{1+\varepsilon'}$ 
\STATE \textbf{break} 
\ENDIF
\ENDFOR
\STATE \label{line:tildetheta} 
	$\tilde{\theta}=\left\lceil \frac{2 w_1\cdot N\cdot T\left(\frac{1}{2} \alpha+\beta\right)^{2}}{L B \cdot \varepsilon^{2}}\right\rceil$ /* $\alpha$ and $\beta$ is defined in Eq.\eqref{eq:alphabeta}.*/
\STATE \label{line:regenerateRRsets} clear the $\cR$, and regenerate $\tilde{\theta}$ PW-RR sets into $\cR$
\STATE $\hat{\cS}^{g}$ = PRM-NodeSelection-OINS$(\cR,k, (w_t)_{t\in[T]})$
\STATE return $\hat{\cS}^{g}$

\end{algorithmic}

\end{algorithm}

The following theorem summarizes the correctness and time complexity of our algorithm PRM-IMM.
\begin{theorem}\label{theorem:main result}
	Let $\hat{\cS}^{g}$ be the output of PRM-IMM-OINS and $\cS^*$ be the optimal solution of round-weighted influence maximization problem. 
	We have that with the probability at least $(1-\frac{1}{N^l})$, 
	\begin{displaymath}
	\rho^{OI}_T(\hat{\cS}^{g}) \geq \left(\frac{1}{2}-\varepsilon \right)\rho^{OI}_T(\cS^{*}),
	\end{displaymath}
	where $\varepsilon > 0$, $l > 0$. 
	The expected running time of PRM-IMM-OINS is $O((k+l)(M+N)T\log (N\cdot T)/\varepsilon^2)$. 
\end{theorem}
The proof of the theorem follows the analysis structure of IMM~\cite{tang15}, and utilizes the key result in Lemma~\ref{lemma:property of PW-RR}.

We can see that PRM-IMM-OINS provides the approximation guarantee while still runs in time near linear to the graph size.
Compared to the original IMM, the running time is increased by a factor of $T$,  caused by the spreading of search of seeds across $T$ rounds rather than in one round.

\section{Results on the NIOS Setting}

We now turn to the non-overlapping influence with overlapping seeds (NIOS) setting.
The NI setting causes important changes to the objective function and the algorithm design, as we will describe in this section.

\subsection{Objective Function under the NIOS Setting}
Similar to Eq.\eqref{eq:ratio plus one(OINS)} for the OINS setting, we have 
\begin{align}\nonumber
    r^{NI}_T(\cS)  = \left(r_0+1\right)\prod_{t=1}^T \left(1+\frac{\sigma^{NI}(S_t)}{d_0^n+d_0^p+z\cdot t+\sum_{i=1}^{t-1} \sigma^{NI}(S_i) }\right)-1.
\end{align}
The main difference is that it uses the marginal influence  $\sigma^{NI}(S_t)$ as defined in Eq.\eqref{eq:influence spread(NI)}, 
	because influence cannot be overlapping under NIOS.
With a similar heuristic simplification, we obtain the round-weighted influence function with $w_t$ defined in \eqref{eq:setweight}:
\begin{align}\label{eq:weighted im NIOS}\nonumber
    \rho^{NI}_T(\cS)=\sum_{t=1}^{T} {w_{t}\cdot\sigma^{NI}(S_t)}.
\end{align}
One can show that $\rho^{NI}_T(\cS)$ has the following equivalent form:
\begin{equation}\label{eq:weighted im NIOS2}
    \rho^{NI}_T(\cS)=\sum_{v\in V}{\E_{L_1,\ldots, L_T} \left[\max_{t\in [T]}w_t \cdot \I\{v\in \Gamma(S_t,L_t)\} \right]}.
\end{equation}
This is because for any node $v\in V$, it may be influenced by some pair $(u,t) \in \cS$ in round $t$, 
	and when multiple seed-round pairs occur, the one with the largest weight among these pairs is taken to calculate the final influence to $v$.
Based on Eq.\eqref{eq:weighted im NIOS2}, one can show that $\rho^{NI}_T(\cS)$ is monotone and submodular.

\subsection{Efficient Algorithm for NIOS setting}
The PRM-IMM-NIOS algorithm follows the same structure as the PRM-IMM-OINS algorithm, but due to the non-overlapping influence requirement and the round-based
	weights nature, the RR set and the node selection procedure has significant difference, as we now explain.
	
For the RR set, for each randomly selected root $v\in V$, we need to generate $T$ RR sets rooted at $v$, denoted as $R_{v,t},t\in [T]$ in pair notation. 
The multi-round RR set (MR-RR set) is the union of all $T$ RR sets, as $R_v = \bigcup_{t=1}^{T} R_{v,t}$.
This is because in the NI setting a node $v$ is only influenced once in all $T$ rounds.
The MR-RR set is similar to the one used in~\cite{sun2018MRIM}, but here different rounds have different weights, leading to a more complicated node selection
procedure.

In the PRM-NodeSelection-NIOS procedure (Algorithm~\ref{alg:nodeselection NIOS}), we maintain a variable $t_R$ for each MR-RR set $R$, and it denotes
	the lowest round number $t$ in which the root $v_R$ of $R$ has been influenced by some already selected seed nodes.
Initially, $t_R = T+1$, since no seed has been selected and thus no node is influenced.
We also use variable $c[(u,t)]$ to store the remaining marginal influence of pair $(u,t)$ on all MR-RR sets.
Variables $t_R$ and $c[(u,t)]$ are used to keep track of the marginal influence to $v_R$ when some additional pair $(v,t)$ is selected as a seed
	(lines~\ref{line:update1}--~\ref{line:update2}):
	when $t < t_R$, this means that the root $v_R$ is now influenced at the lowest round $t$, and thus
	for any other pair $(u,t') \in R$ with $t' < t_R$, if $t' \ge t$, it no long generates any marginal influence to $v_R$
	(the marginal influence is reduced by $w_{t'} - w_{t_R}$, line~\ref{line:zeromarginal});
	and if $t' < t$, the remaining marginal influence $(u,t')$ could generate to $v_R$ through $R$ is $w_{t'} - w_t$
	(the marginal influence is reduced by $w_t-w_{t_R}$, line~\ref{line:positivemarginal}).
This makes this node selection procedure more complicated than the OINS setting.
Note that in line~\ref{line:overlapseeds}, we select the next seed from $V \times[k]\backslash \cS$, reflecting that we allow overlapping seeds in different rounds.

\begin{algorithm}[t]
\caption{PRM-NodeSelection-NIOS}
\label{alg:nodeselection NIOS}
\leftline{\textbf{Input:} the set of MR-RR sets $\cR$, seeds budget $k$,}
\leftline{weights $(w_t)_{t\in[T]}$ (Eq.\eqref{eq:setweight}). particularly, the value of $w_{T+1}$ is $0$.}
\leftline{\textbf{Output:} $k$ size seed set $\hat{\cS}^g$.}
\begin{algorithmic}[1]
\STATE $RR[(v,t)] = \{R\in \cR | (v,t) \in R\}$, $c[(v,t)]= w_t \cdot |RR[(v,t)]|$, $\forall (v,t)\in V\times [T]$ 
/* $RR[(v,t)]$'s and $c[(v,t)]$'s can be constructed during the generation of $\cR$ */
\STATE for all $R\in\cR, t_R = T+1$; $w_{T+1} = 0$
\STATE $\cS=\emptyset$ /* the element of $\cS$ is $(v,t)$ */
\FOR{$i=1$ to $k$}        
    \STATE $(v, t)=\operatorname{argmax}_{(v, t) \in V \times[k]\backslash \cS} \label{line:overlapseeds}
    c[(v,t)]$      
    \STATE $\cS = \cS \cup (v,t)$ 
    \FOR{$R \in RR[(v,t)] \wedge  t < t_R $} \label{line:update1}
        \FOR{$(u,t') \in R \wedge (u,t') \neq (v,t) \wedge t' < t_R$}
            \IF{$t' \geq t$}
                \STATE $c[(u,t')] = c[(u,t')] -(w_{t'}-w_{t_R})$ \label{line:zeromarginal}
            \ELSE
                \STATE $c[(u,t')] = c[(u,t')] -(w_{t} - w_{t_R})$ \label{line:positivemarginal}
            \ENDIF
        \ENDFOR
        \STATE $t_R = t$
    \ENDFOR \label{line:update2}
\ENDFOR
\RETURN $\hat{\cS}^g = \cS$
\end{algorithmic}
\end{algorithm}

In the sampling part, similar to OINS setting, we can iteratively guess the lower bound $x_i$ by halving and verifying through NodeSelection. 
We only need to modify parameter $\beta$ to 
$\sqrt{\frac{1}{2} \cdot\left(\ln \binom{NT}{k}+\alpha^{2}\right)}$,
and the $\binom{NT}{k}$ is due to selecting $k$ pairs from all $NT$ pairs in the overlapping seeds setting.
All other aspects of sampling part remains the same, and thus we ignore the pseudocode here.
For the other two problems mentioned in Section \ref{sec:problem definition}, NINS and OIOS, we can plug in different components in OINS and NIOS to achieve the goal,
	and we omit the details here.

\section{Other Variants of PRM Problem} \label{sec:variants}
In this section, we show that our object function analysis and algorithm design are robust in that they can be adapted to a number of variants
	of the PRM problem.
\subsection{Per-round Budget}
In the PRM definition (Definition~\ref{def:PRM}), the budget $k$ is over all $T$ rounds.
One variant is that every round has its own budget $k$, similar to~\cite{sun2018MRIM}.
It also has four settings as in Table~\ref{Four settings}.
Among them, the OIOS setting is the same as solving the traditional IM problem and then using the same seed set for all rounds.
Other settings can be solved following the similar technique as the overall budget case, but when selecting next node-round pair as a seed, it needs to respect the
	per-round budget rather than the overall budget.
\subsection{Variable Natural Growth Count}
\label{sec:variable z}
We can relax the assumption that the natural growth count $z$ is constant for all rounds to allow a 
	natural growth vector $\bz = [z_1, z_2, \cdots, z_T]$ with $z_t$ denoting the natural customer count in round $t$.
Then, similar to the derivation process in Section~\ref{sec:OINSobj}, 
we can obtain the formula for $r^{OI}_T(\cS)$ as follows:
\begin{equation*}
	r^{OI}_T(\cS)  = \left(r_0+1\right)\prod_{t=1}^T \left(1+\frac{\sigma(S_t)}{d_0^n+d_0^p+\sum_{i=1}^{t}z_i+\sum_{i=1}^{t-1} \sigma(S_i)}\right)-1.
\end{equation*}
The derivation of the surrogate function $\rho^{OI}_T(\cS)$ and the rest algorithm design is the same, and we just need to redefine the weight 
$w_t = 1 / (d_0^n+d_0^p+\sum_{i=1}^{t}z_i)$.

\subsection{Popular Item Engages in Promotion} \label{sec:popularPromotion}
Our PRM problem considers passive popular item that does not engage in promotion.
When popular item also engages in promotion, the situation is more complicated.
Here we assume that the promotion of the popular item in all rounds are known, and the promotional result on round $t$ is given as $p_t$, $t\in [T]$.

The case of unknown popular item promotions would lead the problem into the game-theoretic setting, which is beyond the scope of this paper.
When $p_1, \ldots, p_T$ are known, in the OINS setting, the close form of the objective function $r^{OI}_T(\cS)$ is not easily derived as in Lemma~\ref{lem:OINSobjective}.
Instead we can obtain its upper and lower bound as follows.
\begin{align}
	&r_{T}\le \left(r_{0}+1\right) \prod_{t=1}^{T}\left(1+\frac{\sigma\left(S_{t}\right)}{d_{0}^{n}+d_{0}^{p}+z t+\sum_{i=1}^{t-1} (\sigma(S_{i}) + p_{i})}\right)-1, \label{eq:upperratio}\\
	&r_{T} \ge \left(r_{0}+1\right) \prod_{t=1}^{T}\left(1+\frac{\sigma\left(S_{t}\right)-p_{t}}{d_{0}^{n}+d_{0}^{p}+z t+\sum_{i=1}^{t-1} (\sigma(S_{i})+ p_{i})+p_{t}}\right)-1. \label{eq:lowerratio}
\end{align}
Eq.\eqref{eq:upperratio} can be viewed as treating $p_t$ as the additional natural growth in round $t$, and thus disfavor of the popular item and in favor of the novice item.
Eq.\eqref{eq:lowerratio} works when $d_p^t \ge d_n^t$, i.e. popular item is still the majority in popularity.
These bounds lead to two weight settings:
\begin{equation*}
	\overline{w}_t = \frac{1}{d_0^n+d_0^p+t\cdot z + \sum_{i=1}^t p_i},\quad \underline{w}_t = \frac{1}{d_0^n+d_0^p+t\cdot z+ \sum_{i=1}^t p_i + p_t}.
\end{equation*}
Then following the sandwich approximation strategy~\cite{LCL2016}, we can apply the two sets of weights to the corresponding PRM-IMM algorithm, obtain two seed allocations, and compare their results
	to find the better one.
The approximation ratio for the surrogate objective would have an extra factor of $\max_{t\in T} \underline{w}_t / \overline{w}_t$.

\subsection{Two Other Optimization Objectives}\label{sec:solve two variants}
Beyond the fundamental problem of maximizing the popularity ratio, it is natural to investigate two other variants of the PRM problem in surpassing the competitor:  
minimizing seed budget and minimizing the number of rounds needed. 

\begin{itemize}
    \item Seed Minimization: Given a round budget $T$, minimize the number of seeds in $\cS$ such that the final popularity ratio is at least $1$, i.e., $r_T(\cS)\ge 1$.
    \item Round Minimization: Given a seed budget $k$, minimize the round number $T$ such that the final popularity ratio is at least $1$, i.e., $r_T(\cS)\ge 1$.
\end{itemize}
Note that $r_T(\cS)\ge 1$ means that the novice item has successfully caught up or surpassed the popular item at the end of the promotional period.
For the seed minimization problem, we first find the range of minimum budget by doubling the budget from $1$, and then determine the minimum budget that would make $r_T \ge 1$ by 
	a binary search on the budget, and for each budget tried, we use the PRM-IMM algorithms developed.
Same binary search idea can be applied to the round minimization problem. 

\begin{table}
    \caption{Datasets}
    \label{dataset}
    \begin{tabular}{p{1.2cm}p{0.4cm}p{0.7cm}p{0.8cm}p{0.6cm}p{0.45cm}p{0.55cm}p{0.55cm}}\\
     \toprule[2pt]
      & DM & LastFM & FX01/02 & HEPT  & Yelp & DBLP & LJ \\
     \hline
      nodes& 0.6K & 1.8K & 29.3K & 27.7K & 110K & 650K & 4800K \\
      edges& 3.3K & 12.7K& 212.6K & 352.8K  & 950K & 1900K & 42800K \\
     avg.degree & 4.96 & 13.44 & 7.24 & 12.70 & 15.92 & 3.04 & 8.83 \\
     \bottomrule[2pt]
\end{tabular}
\end{table}

\begin{table}
    \caption{Parameter Settings}
    \label{parameter setting}
    \begin{tabular}{p{0.2cm}p{0.3cm}p{0.7cm}p{0.45cm}p{0.45cm}p{0.55cm}p{0.55cm}p{0.55cm}p{0.6cm}}\\
    \toprule[2pt]
            & DM  & LastFM & FX01 & FX02& HEPT  & Yelp & DBLP & LJ \\
     \hline
      $r_0$ & 0.25 & 0.067 & 0.13  & 0.067 & 0.25 & 0.02 & 0.025  & 0.025 \\
      $d_0$ & 900  & 1600 & 1700  & 1600 & 5000 & 51000    & 20500 & 205000 \\
      $z$   & 10   & 100  & 50    & 100 & 150 & 2000     & 2500 & 25000 \\
     \bottomrule[2pt]
\end{tabular}
\end{table}

\subsection{ Multiple Items}
\label{sec:multiItem}
In this section, we discuss the PRM problem when there are multiple items in the system.

\textbf{The Novice Item View Point, with Multiple Promotional Items}
We first focus on the view point from the targeted novice item, which is trying to increase its popularity compared with other items. Other competing items may or may not have their own promotions. 
We denote other items as items $1, \ldots, s$, and the popularity of item $j\in [s]$ at the end of round $t$ is $d_t^j$.

We do not need to distinguish between popular and novice items among other items, and only need to denote whether an item conducts promotion or not. 
Each item $j \in [s]$ may conduct promotion $p_t^j$ in round $t$, and $p_t^j=0$ means no promotion.
The popularity ratio of our targeted novice item, $r_t$, is defined as $r_t = d_t^n /(\sum_{j=1}^s d_t^j)$. 
The following theorem demonstrates that, even in the multiple-item case, we are able to derive the upper and lower bounds for the popularity ratio $r_T$, similar to the results presented in Section~\ref{sec:popularPromotion}.
\begin{theorem}\label{theorem:multi-item promotion}
	In the setting of multiple items with promotions, we can bound $r_T$ as follows: 
	\begin{align}
	& r_T \le \left(r_{0}+1\right) \prod_{t=1}^{T}\left(1+\frac{\sigma\left(S_{t}\right)}{\zeta}\right)-1, \label{eq:multi-item upperratio}\\ 
	& r_T \ge \left(r_{0}+1\right) \prod_{t=1}^{T}\left(1+\frac{\sigma\left(S_{t}\right)-\sum_{j=1}^{s}p_{t}^{j}}{\zeta+\sum_{j=1}^{s}p_{t}^{j}}\right)-1, \label{eq:multi-item lowerratio}
	\end{align}
 where $\zeta = d_{0}^{n}+\sum_{j=1}^{s}(d_{0}^{j})+z t+\sum_{i=1}^{t-1} (\sigma(S_{i}) + \sum_{j=1}^{s}p_{i}^{j})$,
 	and Inequality~\eqref{eq:multi-item lowerratio} holds when $d_t^n \le \sum_{j=1}^s d_t^j$, for $t <T$. 
\end{theorem}

The implication of the above theorem is that, from the perspective of the targeted novice item, all other items can be treated as one single composite item ---
	as long as we know the combined promotional effect of all other items, namely $\sum_{j=1}^{s}p_{t}^{j}$, we can treat all other items as one single item
	and compute the promotion to the novice item in the same way as described in Section~\ref{sec:popularPromotion}.
Moreover, it is easy to verify that if none of the other items conduct promotion, i.e. $p_t^j=0$ for all $t$'s and $j$'s, we can treat all other items as
	a single composite item and derive the exact formula for $r_T$ as in Lemma~\ref{lem:OINSobjective}. 
The main reason for the above property is that the preferential attachment mechanism remains the proportional allocation whether the novice item is facing one other item or many other items.

This further extends our understanding of popularity ratio maximization in cases where multiple items (popular or not, having promotions or not) co-exist in the system.

\textbf{The Game Theoretic Point of View}
	The above is from one novice item's view point, on how to allocate seeds to maximize its popularity ratio when other items' promotions are known.
	When items' promotions are unknown to one another, we need to study the dynamic of the overall system, and we enter the game-theoretic setting.
	Here, each item's strategy is the allocation of seeds over $T$ rounds under its budget, and its utility is its popularity ratio.
	The algorithm we designed before can be viewed as the algorithm for each individual item (player) to achieve its best response, or approximate best response.
	Collectively, we can study the best-response dynamic and the equilibrium behavior.
	Alternatively, one may also map our basic PA-IC model into the Stackelberg game model~\cite{FT93}, where the novice item and the popular make moves one after another based on the other side's previous moves. 
	However, the full study of the game-theoretic aspect under our PA-IC model and the PRM problem setting
	would be too complicated to be included in this paper, and we will leave the detailed formulation and the technical study as a future
	work item.

\section{Experiments}
In this section, we empirically validate our algorithms on eight real-world networks.

\label{sec:experiments}
\subsection{Experiment Setup}

\textbf{Data description}
We apply eight real-world network datasets in our experiments, with basic statistics shown in Table \ref{dataset}.
The DM dataset is a collaboration network where nodes represent the data-mining researchers and edges represent co-authorships~\cite{TangSWY09}. 
The FX01/02 datasets are from the Flixster, a network of American social movie discovery services \cite{barbieri2012topic}. 
Each node is a user and a directed edge from node $u$ to $v$ is formed if $v$ rates one movie shortly after $u$ does so on the same movie.
FX01/02 represent two networks on two different movie topics~\cite{barbieri2012topic}, 
The NetHEPT dataset is an academic collaboration network from the “High Energy Physics Theory” section of arXiv from 1991 to 2003, where nodes represent the authors and each edge represents one paper co-authored by two nodes~\cite{ChenWW10}. 
LastFM is a UK-based music site that focuses on online radio and music communities \footnote{http://www.last.fm/}. We use the friend relationship on LastFM to build a user relationship social network.
The Yelp dataset is released by yelp officially\footnote{http://www.yelp.com}. Yelp.com is the largest review site in the United States. We use the friend relationship data in Yelp to build a user's social network.
DBLP is another academic collaboration network extracted from the online archive DBLP (dblp.uni-trier.de) and used for influence studies \cite{WCW12}.
LiveJournal, denoted as LJ, is a free online community with almost 10 million members; a significant fraction of these members are highly active. It is obtained from Stanford’s SNAP project \cite{snapnets}. 

\begin{figure*}
	\centering
	\includegraphics[scale=0.32]{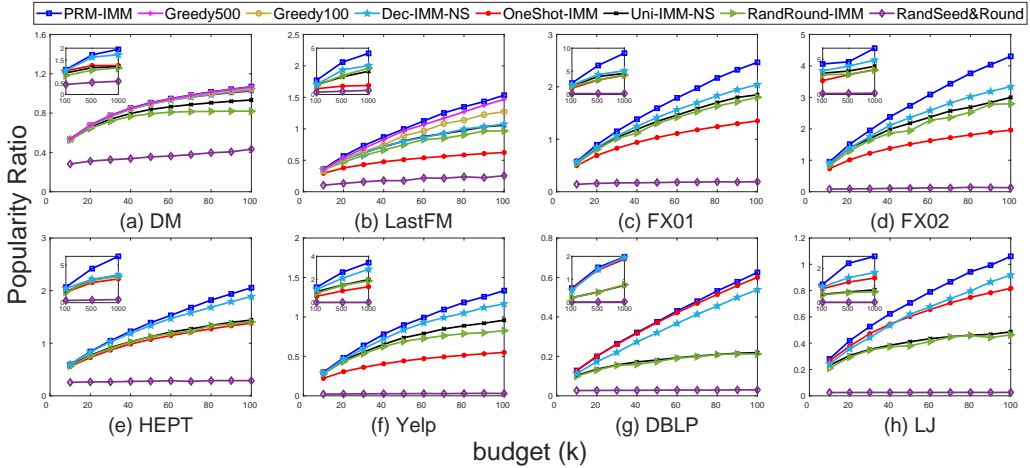}
	\caption{The popularity ratio vs. budgets for different algorithms in OINS setting.}
	\label{fig:ratio}
\end{figure*}

\textbf{Metrics}
We consider the following three metrics, reflecting what one could consider in practice:

\begin{itemize}[leftmargin = 9pt]

    \item \textbf{Popularity ratio.}
    The popularity ratio at the end of promotions.
    \item \textbf{Surpass budget.} Given a fixed round $T$, 
    	the minimum budget needed for the novice item to surpass the popular item.
    \item \textbf{Surpass time.} Given a fixed budget $k$, 
    	the minimum round number when the novice item surpasses the popular item.
\end{itemize}

\textbf{Parameter setup and reproducibility}
For DM and FX datasets, we use the edge weights learned in their respective studies~\cite{TangSWY09,barbieri2012topic}.
For the remaining datasets, we use the following “Weighted Cascade” policy adopted in \cite{kempe03} and often used in other influence maximization studies (e.g. \cite{ChenWY09,WCW12,tang15}): 
for each edge $(u,v)$, $p(u,v)$ is set to be $\frac{1}{N_{in}(v)}$, where $N_{in}(v)$ is the in-degree of $v$. 
Every PRM instance is also determined by the parameters $r_0,d_0,z$.
We set proper parameters for each dataset, as shown in Table~\ref{parameter setting}, to reflect different scenarios and match with the network size. Specifically, the settings of $d_0$ and $z$ are mainly related to the size of the network. When the network size is large, promotion may generate large popularity from the social network. Thus, to match the natural growth parameter and to avoid promotion dominating the natural growth, we need to adjust $d_0$ and $z$ accordingly to compensate for the large popularity generated by promotion. Parameter $r_0$ further covers different scenarios among the datasets.

We further vary these parameters and study their effects (Fig.\ref{fig:parameter analysis popularity ratio} and Fig.\ref{fig:parameter analysis max round}). 
Following \cite{tang15}, we set the error probability parameter $l = 1$, approximation ratio parameter $\varepsilon = 0.1$ in PRM-IMM.
For implementation details see the source code and proof appendix in this repository\footnote{https://github.com/Complex-data/PRM}.

\textbf{Baselines}
There is no existing baseline for the newly proposed PRM problem. Thus, we compare with the following reasonable
	heuristic baselines.

\begin{itemize}[leftmargin = 9pt]
\item \textbf{OneShot-IMM}: The largest $k$ nodes are picked by the IMM algorithm, and all are placed in round $1$. 

\item \textbf{Uni-IMM-NS/Uni-IMM-OS}: The NS version  selects the largest $k$ nodes by IMM and places them in all rounds separately with
	equal size in the order of the greedy selection. The OS version simply repeats the first round seeds in later rounds (for NIOS tests).

\item \textbf{RandRound-IMM}: Pick the largest $k$ nodes by the IMM algorithm and put them into $T$ rounds randomly. 

\item \textbf{RandSeed\&Round}: Randomly select $k$ nodes and randomly place these nodes in each round.

\item \textbf{Dec-IMM-NS/Dec-IMM-OS}: The NS version picks the largest $k$ nodes by IMM and places seed nodes in descending order. The method of generating the decreasing sequence is to use one-fifth of the budget in the first round and one-fifth of the remaining budget in each subsequent round. The OS version repeats the seeds in the first round with the same budget as the NS version.
\item \textbf{Greedy}: Directly apply the greedy algorithm, using simulations to estimate marginal gains and do not use the RIS approach.

\end{itemize}

The first five baselines are all heuristics with no theoretical guarantee.
Algorithm Greedy has the theoretical guarantee as PRM-IMM, but it is very slow in running simulations, and we will only show its result on the two small datasets DM and LastFM.

\subsubsection{Platform}

We implement all algorithms in C++, compiled in Visual Studio 2019, and run our tests on a computer with 3.6GHz Intel(R) Core(TM) i9-9900K CPU (16 cores), 128G memory, and Windows 10 professional (64 bits).

\subsection{Results}
\textbf{OINS main results.} \label{sec:PRM main results}
Fig.~\ref{fig:ratio} shows the result of popularity ratio vs. budget on all datasets and all algorithms in the OINS setting.
The results are averaged over 10 independent runs.
Note that the y-axis corresponds to the original objective function (Eq. \eqref{eq:ratio plus one(OINS)}), not the surrogate function (Eq. \eqref{eq:weighted im}).
It is obvious that PRM-IMM outperforms the baseline algorithms with a large margin in many cases. 
More specifically, RandomSeed\&Round is much worse than other algorithms, 
	showing that randomly giving out promotions should be avoided if possible.
Uniform-IMM and RandomRound-IMM have similar performance since both spread the seeds more or less evenly among the $T$ rounds, and in many cases such even spreading results in poor performance.
OneShot-IMM is at the opposite extreme, only allocating seeds in the first round. 
It has unstable performance, and in most cases it performs poorly, except for the DM and DBLP datasets.

Dec-IMM in general performs better than the OneShot/Uniform\\/RandomRound versions, but it still has a clear gap with our algorithm PRM-IMM.
This is because all these heuristics cannot adjust to different network structure and diffusion characteristics of the networks (such as whether important seeds are likely to have overlapping influence), while our PRM-IMM automatically balances the early promotion and the influence redundancy among seeds.

\begin{table}
	\caption{Running Time (OINS setting, $T=100$, $k=100$, results in seconds, average over 5 runs)}
	\label{tab:time(OINS)}
	\begin{tabular}{c|cccc}
		\toprule[2pt]
		& IMM & PRM-IMM & Greedy100 & Greedy500 \\
		\hline
		DM& 0.067 & 13.7 & 397.8& 3489 \\
		LastFM& 1.24  & 82.86 & 11976 & 53681 \\
		FX01& 4.07  & 437 & - & - \\
		FX02& 3.42  & 432 & - & - \\
		HEPT& 2.39  & 399 & - & - \\
		Yelp& 44.36 & 2959 & - & - \\
		DBLP& 47.46  & 1176 & - & - \\
		LJ& 371  & 15026 & - & - \\
		\bottomrule[2pt]
	\end{tabular}
\end{table}

\begin{table}
	\caption{Running Time (OINS setting, $k=500,1000$ results in seconds, average over 5 runs)}
	\label{tab:time(large k)}
	\begin{tabular}{c|cc|cc}
		\toprule[2pt]
  & \multicolumn{2}{c|}{$k=500$}  & \multicolumn{2}{c}{$k=1000$}  \\
        \hline
		& IMM & PRM-IMM & IMM & PRM-IMM \\
		\hline
		DM& 0.13 & 1.16 & 0.13& 1.98 \\
		LastFM& 0.54  & 48.59 & 0.61 & 48.59 \\
		FX01& 11.47  & 138.62 & 16.51 & 229.9 \\
		FX02& 10.52  & 137.52 & 16.25 & 253.48 \\
		HEPT& 15.11  & 67.83 & 13.79 & 106.93 \\
		Yelp& 20.36 & 1639.34 & 24.51 & 2715 \\
		DBLP& 332.46  & 1317.18 & 383.73 & 1939.36 \\
		LJ& 2505  & 5607 & 4251 & 10760 \\
		\bottomrule[2pt]
	\end{tabular}
\end{table}

Greedy should perform equally well as PRM-IMM, since PRM-IMM is just a more efficient way of implementing Greedy.
But Greedy is very time consuming --- as shown in Table \ref{tab:time(OINS)}, Greedy with 100 and 500 simulations per estimation is already several orders of magnitude slower than PRM-IMM, and it takes Greedy500 almost 15 hours to run on the 1.8K node LastFM network.
Thus Greedy is not scalable and we cannot run it on other datasets except for DM and LastFM.
Even for LastFM, Fig.~\ref{fig:ratio}(b) shows that Greedy500 and Greedy100 are not as competitive as PRM-IMM, due to the insufficient number of simulations.

The top left corner of each subfigure in Fig.\ref{fig:ratio} shows the result of different algorithms under larger seed  budget of $k=500,1000$. 
	It can be seen that PRM-IMM still performs better than other baselines for larger $k$'s.

Table \ref{tab:time(OINS)} shows that the running time of the algorithms with $k=100$. 
Note that Random Seed\&Round is a trivial baseline with almost instant running time, but its quality is far worse than all others as reported in Fig.\ref{fig:ratio},
	and thus we do not include it in the running time table.
For all other baselines except Greedy, they call IMM to select seeds first, and then do various simple allocation of seeds over rounds.
Thus their running time is dominated by the IMM running time, which is reported in the table.
The result shows that (a) Greedy100/Greedy500 is two to three orders of magnitude slower than PRM-IMM, and thus is not scalable at all;
(b) PRM-IMM is around 100 times slower than the pure IMM algorithm (except for DBLP and LJ), which aligns with our theoretical analysis showing that our running time has an extra $T$ factor.
Note that for the two large datasets DBLP and LJ, we reduce its round input $T$ to save time, since the results from smaller budgets show that only the first few rounds get allocated for seeds.

Table \ref{tab:time(large k)} shows the running time of PRM-IMM vs. IMM with $k=500,1000$. 
For this test, we optimize for $T$, meaning that we set $T$ to a small value, and if the resulting allocation of PRM-IMM does not use all rounds, we stop
	the algorithm; otherwise, we double the $T$ value and retry the algorithm. 
We can see that PRM-IMM can still scale to larger seed budgets and large networks.

Overall, our OINS results clearly show that our PRM-IMM algorithm
	outperforms all simple heuristics, while it runs much more efficiently than the greedy algorithm.
PRM-IMM achieves a good balance between quality and efficiency. 

\begin{figure}[t]
	\centering
	\includegraphics[scale=0.4]{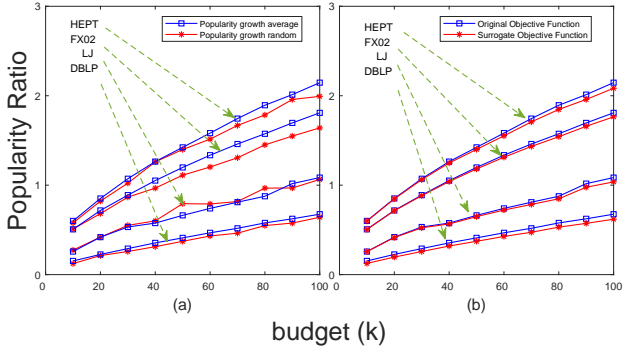}
	\caption{Model justification. (a) Comparing between the random growth (red) and average growth (blue); (b) Comparing between the original objective
		function and the surrogate function.}
	\label{fig:simulation_difference}
\end{figure}

\textbf{Model justification} \label{sec:modeljustify}
We now justify two choices we made in our model formulation.
First, Eq.\eqref{eq:popularityeq} is based on the average popularity growth in each round rather than the random popularity growth, as we discussed in Section \ref{sec:pa-ic model}.
We compare the effect of the average growth with the random growth, using the seed allocation selected by PRM-IMM with the budget ranging from $10$ to $100$.
Fig. \ref{fig:simulation_difference}(a) shows that the results of the two effects are very close in four datasets (the other four have similar results), 
justifying our use of the average effect in Eq.\eqref{eq:popularityeq}.

Next, we compare the original objective function of Eq. \eqref{eq:ratio plus one(OINS)} with the surrogate round-weighted influence function of Eq. \eqref{eq:weighted im}. 
We use seed allocations selected by PRM-IMM for this comparison.
Fig. \ref{fig:simulation_difference}(b) shows that the two functions agree with each other very well, with average derivation less than 5\% for four datasets (and the other four have
	similar results). 
Hence, using the surrogate function is a reasonable choice, even though our PRM is a heuristic with respect to the original objective.

\begin{figure}
	\centering
	\includegraphics[scale=0.40]{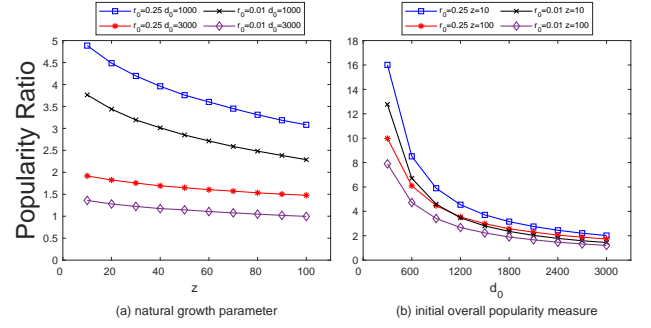}
	\caption{Parameter analysis for PRM under FX01.}
	\label{fig:parameter analysis popularity ratio}
\end{figure}

\begin{figure}
	\centering
	\includegraphics[scale=0.40]{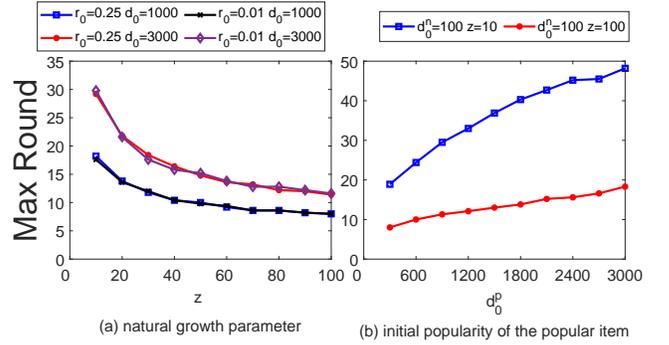}
	\caption{Different parameter's impact on seeds allocation, in the FX01 dataset.}
	\label{fig:parameter analysis max round}
\end{figure}

\textbf{Parameter analysis in PRM}\label{sec:Parameter analysis in PRM}
\begin{figure*}[t]
	\centering
	\includegraphics[scale=0.32]{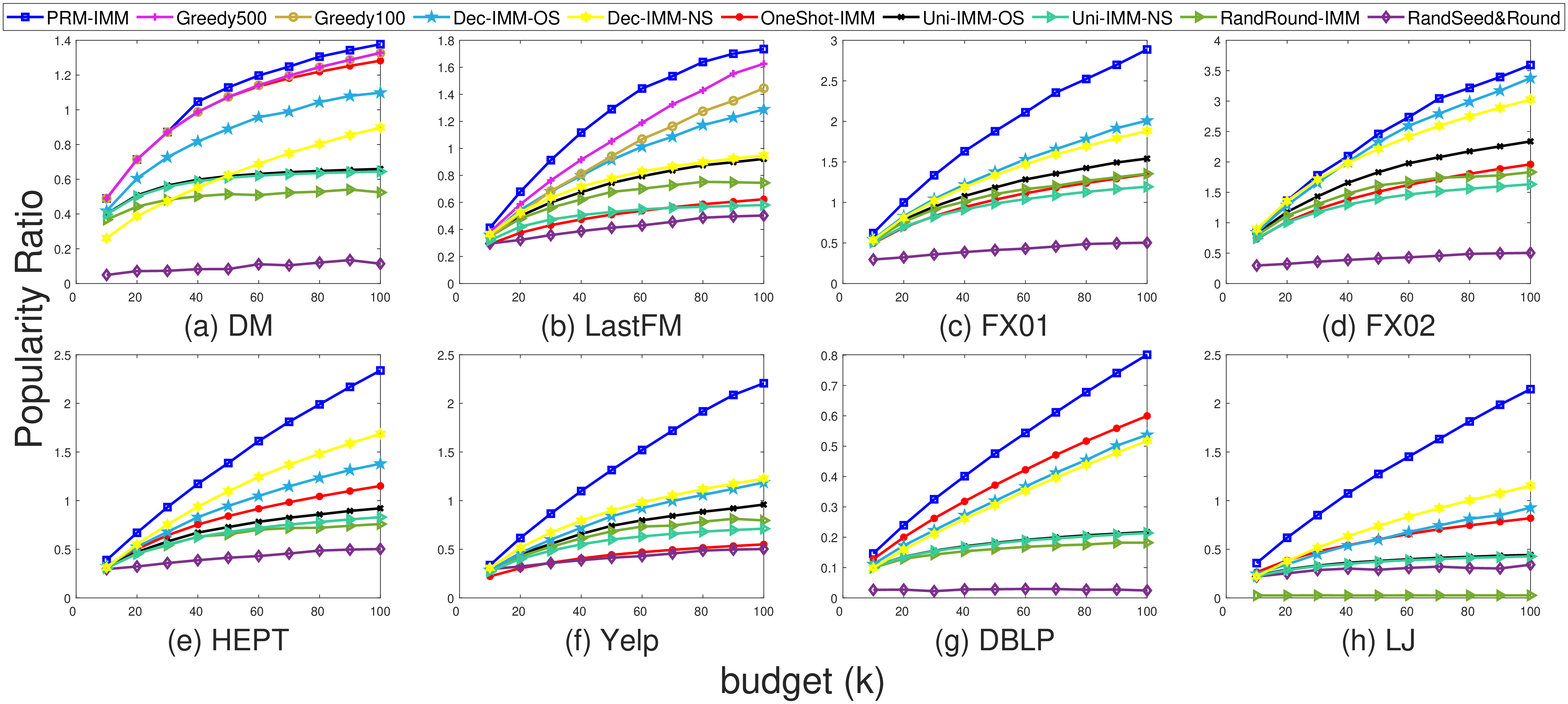}
	\caption{The popularity ratio vs. budgets for different algorithms in NIOS setting}
	\label{fig:ratio NIOS}
\end{figure*}
We now analyze the effect of the initial total popularity measure $d_0$ and natural growth count $z$ on the PRM-IMM algorithm.
We run PRM-IMM to select 100 seeds for this analysis.
Fig. \ref{fig:parameter analysis popularity ratio} shows the result for FX01, and results on other datasets are similar.

Fig. \ref{fig:parameter analysis popularity ratio}(a) shows that when $z$ increases, the popularity ratio decreases, indicating that
	the larger the natural growth of the popularity measure, the more difficult it is for the novice item to catch up with the popular item. 
Fig. \ref{fig:parameter analysis popularity ratio}(b)  also shows the decreasing trend with respect to $d_0$.
This is because when $d_0$ is larger while $r_0$ remains the same, the gap between the popular and novice item becomes larger and it is more difficult for the noice item to catch up.

\begin{table}
	\caption{Running Time (NIOS setting, $T=100$, $k=100$, results in seconds, average over 5 runs)}
	\label{tab:time(NIOS)}
	\begin{tabular}{ccccc}\\
		\toprule[2pt]
		& IMM & PRM-IMM & greedy-100 & greedy-500 \\
		\hline
		DM& 0.067 & 342.8 & 1274.8& 3489 \\
		LastFM& 1.24  & 1258 & 51034 & 249656 \\
		FX01& 4.07  & 4396 & - & - \\
		FX02& 3.42  & 6163 & - & - \\
		HEPT& 2.39  & 8056 & - & - \\
		Yelp& 44.36 & 39582 & - & - \\
		DBLP& 47.46  & 1395 & - & - \\
		LJ& 371  & 76396 & - & - \\
		\bottomrule[2pt]
	\end{tabular}
\end{table}

Among the parameters, $z$ and $d^p_0$ (initial popularity of the popular item) are not directly related to the novice item.
Thus, we further investigate how $z$ and $d^p_0$ would affect the seed allocation choice of the novice item in Fig. \ref{fig:parameter analysis max round}.
Y-axis ``max round'' means the maximum number of rounds used in the allocation.
Fig. \ref{fig:parameter analysis max round} shows that when $z$ increases, the max round decreases, and when $d^p_0$ increases, the max round also increases.
This shows the interesting adaptive property of PRM-IMM: when the natural growth count per round increases, the algorithm tries to put more seeds earlier, otherwise natural growth in each round
	will cause more people going to the popular item and making it harder to catch up; but if only the initial popularity of the popular item increases, the algorithm can use more rounds to achieve a
	better catch-up result in the end.
\begin{figure*}[t]
	\centering
	\includegraphics[scale=0.25]{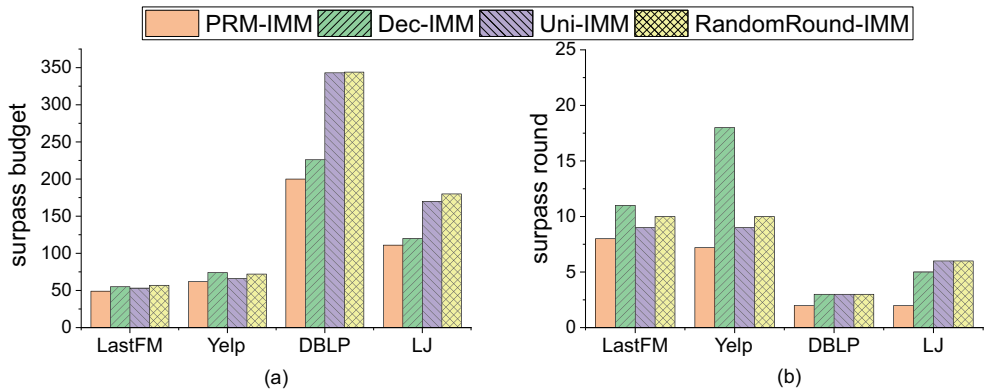}
	\caption{(a) Fix round $T=20$, find the smallest budget to make the popularity ratio exceed $1$. 
		(b) Fix budget (LastFM, $k=55$; Yelp, $k=75$; DBLP, $k=230$;  LJ, $k=125$), make the popularity ratio exceed $1$ in as few rounds as possible.} 
	\label{fig:beating}
\end{figure*}

\textbf{Results on NIOS}
The popularity ratio and running time results for NIOS are shown in Fig.~\ref{fig:ratio NIOS} and Table~\ref{tab:time(NIOS)}, respectively.
For this test, we further include baselines Uni-IMM-OS and Dec-IMM-OS to consider seed overlaps.
Overall, the results are similar to those of OINS, and our PRM-IMM algorithm always performs the best among all baselines, and it runs significantly faster than the simple greedy algorithm.
The performance gap between our algorithm and other baselines is even larger in many cases compared to OINS.
This is perhaps because the NIOS setting demands more sophisticated seed allocation. 
For example, we notice that seed overlaps vary with the networks.
The OS specific baselines in many cases perform not as well as their NS counterparts, indicating that it is not easy to adjust proper overlaps among seeds.
Compared to OINS, our PRM-IMM algorithm runs slower, and this is mainly because each MR-RR set for NIOS is much larger than the PW-RR set for OINS.

\begin{figure*}
	\centering
	\includegraphics[scale=0.30]{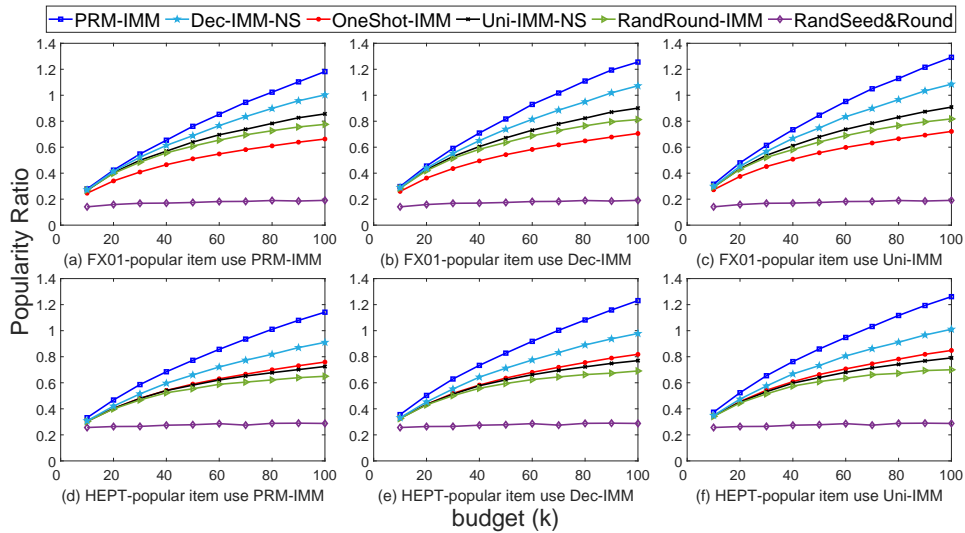}
	\caption{ The popularity ratio vs. budgets for different algorithms when the popular item has promotion. 
 }
	\label{fig:ratio popular promotion}
\end{figure*}
\textbf{Alternative Objectives.}\label{sec:PRM variants}
Finally, we evaluate the two alternative objectives discussed in Section \ref{sec:solve two variants}. 
Since RandomSeed\&Round is clearly worse than other algorithms, while OneShot-IMM cannot assign seeds to multiple rounds, we do not include them as baselines in this part.
We apply the idea in Section \ref{sec:solve two variants} to conduct these experiments for the OINS setting.

Fig.\ref{fig:beating} shows the results on four datasets, and the results on the other four datasets are similar.
Fig.\ref{fig:beating}(a) shows the result for the Seed Minimization problem, 
    in which we fix $T=20$, and then use different algorithms to find the minimum budget that can make the popular ratio exceed 1. 
Fig.\ref{fig:beating}(b) shows the result for the Round Minimization problem.
For this problem, we need to set different budgets for different datasets based on our result from Fig.\ref{fig:beating}(a) to balance the influence of $T$ and $k$.

The results in both figures show that our algorithm can achieve the best results compared to all other baselines.

\textbf{ Popular item having promotion.}\label{sec:PRM variants}
In this section, we evaluate different algorithms in the situation when the novice item  knows the promotion plan of the popular item.
We assume that the popular item uses a budget of $30$ and uses three different strategies: PRM-IMM, Dec-IMM, Uni-IMM, to select seeds. Then the novice item uses PRM-IMM and other baselines to select seed nodes and evaluate the final popularity ratio. The results on FX01 and HEPT dataset are shown in Fig.\ref{fig:ratio popular promotion}.
We can see that even when the popular item has promotions, our PRM-IMM still performs well and beats all the baselines.
Moreover, the end popularity ratio increases when the popular item promotion moves from PRM-IMM algorithm to Dec-IMM and then Uni-IMM, suggesting that 
	the novice item can achieve better promotion result when the popular item uses a less competitive promotion strategy.

\section{Future Work}

One interesting future direction is to study adaptive strategies, where the seed selection of later rounds depends on the actual promotional results of the earlier rounds.
Another direction is to consider the delayed effect of promotion, where social network promotion may take multiple rounds to generate popularity growth.
One may also look into the model where the promotion of the popular item and the novice item also competes in the social network.



\begin{thebibliography}{36}


\ifx \showCODEN    \undefined \def \showCODEN     #1{\unskip}     \fi
\ifx \showDOI      \undefined \def \showDOI       #1{#1}\fi
\ifx \showISBNx    \undefined \def \showISBNx     #1{\unskip}     \fi
\ifx \showISBNxiii \undefined \def \showISBNxiii  #1{\unskip}     \fi
\ifx \showISSN     \undefined \def \showISSN      #1{\unskip}     \fi
\ifx \showLCCN     \undefined \def \showLCCN      #1{\unskip}     \fi
\ifx \shownote     \undefined \def \shownote      #1{#1}          \fi
\ifx \showarticletitle \undefined \def \showarticletitle #1{#1}   \fi
\ifx \showURL      \undefined \def \showURL       {\relax}        \fi
\providecommand\bibfield[2]{#2}
\providecommand\bibinfo[2]{#2}
\providecommand\natexlab[1]{#1}
\providecommand\showeprint[2][]{arXiv:#2}

\bibitem[Barab{\'a}si and Albert(1999)]%
        {BA99}
\bibfield{author}{\bibinfo{person}{Albert-L{\'a}szl{\'o} Barab{\'a}si} {and}
  \bibinfo{person}{R{\'e}ka Albert}.} \bibinfo{year}{1999}\natexlab{}.
\newblock \showarticletitle{Emergence of Scaling in Random Networks}.
\newblock \bibinfo{journal}{\emph{Science}} (\bibinfo{year}{1999}).
\newblock


\bibitem[Barbieri et~al\mbox{.}(2012)]%
        {barbieri2012topic}
\bibfield{author}{\bibinfo{person}{Nicola Barbieri}, \bibinfo{person}{Francesco
  Bonchi}, {and} \bibinfo{person}{Giuseppe Manco}.}
  \bibinfo{year}{2012}\natexlab{}.
\newblock \showarticletitle{Topic-Aware Social Influence Propagation Models}.
  In \bibinfo{booktitle}{\emph{ICDM}}.
\newblock


\bibitem[Borgs et~al\mbox{.}(2014)]%
        {BorgsBrautbarChayesLucier}
\bibfield{author}{\bibinfo{person}{Christian Borgs}, \bibinfo{person}{Michael
  Brautbar}, \bibinfo{person}{Jennifer~T. Chayes}, {and}
  \bibinfo{person}{Brendan Lucier}.} \bibinfo{year}{2014}\natexlab{}.
\newblock \showarticletitle{Maximizing Social Influence in Nearly Optimal
  Time}. In \bibinfo{booktitle}{\emph{SODA}}.
\newblock


\bibitem[Budak et~al\mbox{.}(2011)]%
        {BAA11}
\bibfield{author}{\bibinfo{person}{Ceren Budak}, \bibinfo{person}{Divyakant
  Agrawal}, {and} \bibinfo{person}{Amr El~Abbadi}.}
  \bibinfo{year}{2011}\natexlab{}.
\newblock \showarticletitle{Limiting the Spread of Misinformation in Social
  Networks}. In \bibinfo{booktitle}{\emph{WWW}}.
\newblock


\bibitem[C{\u{a}}linescu et~al\mbox{.}(2011)]%
        {CalinescuCPV11}
\bibfield{author}{\bibinfo{person}{Gruia C{\u{a}}linescu},
  \bibinfo{person}{Chandra Chekuri}, \bibinfo{person}{Martin P{\'{a}}l}, {and}
  \bibinfo{person}{Jan Vondr{\'{a}}k}.} \bibinfo{year}{2011}\natexlab{}.
\newblock \showarticletitle{Maximizing a Monotone Submodular Function Subject
  to a Matroid Constraint}.
\newblock \bibinfo{journal}{\emph{{SIAM} J. Comput.}} (\bibinfo{year}{2011}).
\newblock


\bibitem[Chen(2018)]%
        {chen18}
\bibfield{author}{\bibinfo{person}{Wei Chen}.} \bibinfo{year}{2018}\natexlab{}.
\newblock \showarticletitle{An Issue in the Martingale Analysis of the
  Influence Maximization Algorithm {IMM}}. In \bibinfo{booktitle}{\emph{Lecture
  Notes in Computer Science}}.
\newblock


\bibitem[Chen(2020)]%
        {chen20}
\bibfield{author}{\bibinfo{person}{Wei Chen}.} \bibinfo{year}{2020}\natexlab{}.
\newblock \bibinfo{booktitle}{\emph{Network Diffusion Models and Algorithms for
  Big Data (in Chinese)}}.
\newblock \bibinfo{publisher}{Post and Telecom Press}.
\newblock


\bibitem[Chen et~al\mbox{.}(2013)]%
        {chen2013information}
\bibfield{author}{\bibinfo{person}{Wei Chen}, \bibinfo{person}{Laks V.~S.
  Lakshmanan}, {and} \bibinfo{person}{Carlos Castillo}.}
  \bibinfo{year}{2013}\natexlab{}.
\newblock \bibinfo{booktitle}{\emph{Information and Influence Propagation in
  Social Networks}}.
\newblock


\bibitem[Chen et~al\mbox{.}(2010a)]%
        {ChenWW10}
\bibfield{author}{\bibinfo{person}{Wei Chen}, \bibinfo{person}{Chi Wang}, {and}
  \bibinfo{person}{Yajun Wang}.} \bibinfo{year}{2010}\natexlab{a}.
\newblock \showarticletitle{Scalable Influence Maximization for Prevalent Viral
  Marketing in Large-Scale Social Networks}. In
  \bibinfo{booktitle}{\emph{KDD}}.
\newblock


\bibitem[Chen et~al\mbox{.}(2009)]%
        {ChenWY09}
\bibfield{author}{\bibinfo{person}{Wei Chen}, \bibinfo{person}{Yajun Wang},
  {and} \bibinfo{person}{Siyu Yang}.} \bibinfo{year}{2009}\natexlab{}.
\newblock \showarticletitle{Efficient Influence Maximization in Social
  Networks}. In \bibinfo{booktitle}{\emph{KDD}}.
\newblock


\bibitem[Chen et~al\mbox{.}(2010b)]%
        {ChenYZ10}
\bibfield{author}{\bibinfo{person}{Wei Chen}, \bibinfo{person}{Yifei Yuan},
  {and} \bibinfo{person}{Li Zhang}.} \bibinfo{year}{2010}\natexlab{b}.
\newblock \showarticletitle{Scalable Influence Maximization in Social Networks
  under the Linear Threshold Model}. In \bibinfo{booktitle}{\emph{ICDM}}.
\newblock


\bibitem[Domingos and Richardson(2001)]%
        {domingos01}
\bibfield{author}{\bibinfo{person}{Pedro Domingos} {and} \bibinfo{person}{Matt
  Richardson}.} \bibinfo{year}{2001}\natexlab{}.
\newblock \showarticletitle{Mining the Network Value of Customers}. In
  \bibinfo{booktitle}{\emph{KDD}}.
\newblock


\bibitem[Fudenberg and Tirole(1993)]%
        {FT93}
\bibfield{author}{\bibinfo{person}{D. Fudenberg} {and} \bibinfo{person}{J.
  Tirole}.} \bibinfo{year}{1993}\natexlab{}.
\newblock \bibinfo{booktitle}{\emph{Game Theory}}.
\newblock \bibinfo{publisher}{MIT Press}.
\newblock


\bibitem[Goyal et~al\mbox{.}(2011)]%
        {simpath}
\bibfield{author}{\bibinfo{person}{Amit Goyal}, \bibinfo{person}{Wei Lu}, {and}
  \bibinfo{person}{Laks V.~S. Lakshmanan}.} \bibinfo{year}{2011}\natexlab{}.
\newblock \showarticletitle{{SIMPATH:} An Efficient Algorithm for Influence
  Maximization under the Linear Threshold Model}. In
  \bibinfo{booktitle}{\emph{ICDM}}.
\newblock


\bibitem[Goyal and Kearns(2012)]%
        {goyal12stoc}
\bibfield{author}{\bibinfo{person}{Sanjeev Goyal} {and}
  \bibinfo{person}{Michael Kearns}.} \bibinfo{year}{2012}\natexlab{}.
\newblock \showarticletitle{Competitive Contagion in Networks}. In
  \bibinfo{booktitle}{\emph{STOC}}.
\newblock


\bibitem[He et~al\mbox{.}(2012)]%
        {HeSCJ12}
\bibfield{author}{\bibinfo{person}{Xinran He}, \bibinfo{person}{Guojie Song},
  \bibinfo{person}{Wei Chen}, {and} \bibinfo{person}{Qingye Jiang}.}
  \bibinfo{year}{2012}\natexlab{}.
\newblock \showarticletitle{Influence Blocking Maximization in Social Networks
  under the Competitive Linear Threshold Model}. In
  \bibinfo{booktitle}{\emph{SDM}}.
\newblock


\bibitem[Jung et~al\mbox{.}(2012)]%
        {JungHC12}
\bibfield{author}{\bibinfo{person}{Kyomin Jung}, \bibinfo{person}{Wooram Heo},
  {and} \bibinfo{person}{Wei Chen}.} \bibinfo{year}{2012}\natexlab{}.
\newblock \showarticletitle{{IRIE:} Scalable and Robust Influence Maximization
  in Social Networks}. In \bibinfo{booktitle}{\emph{ICDM}}.
\newblock


\bibitem[Kempe et~al\mbox{.}(2003)]%
        {kempe03}
\bibfield{author}{\bibinfo{person}{David Kempe}, \bibinfo{person}{Jon
  Kleinberg}, {and} \bibinfo{person}{\'{E}va Tardos}.}
  \bibinfo{year}{2003}\natexlab{}.
\newblock \showarticletitle{Maximizing the Spread of Influence through a Social
  Network}. In \bibinfo{booktitle}{\emph{KDD}}.
\newblock


\bibitem[Kim et~al\mbox{.}(2013)]%
        {kim2013scalable}
\bibfield{author}{\bibinfo{person}{Jinha Kim}, \bibinfo{person}{Seung{-}Keol
  Kim}, {and} \bibinfo{person}{Hwanjo Yu}.} \bibinfo{year}{2013}\natexlab{}.
\newblock \showarticletitle{Scalable and parallelizable processing of influence
  maximization for large-scale social networks?}. In
  \bibinfo{booktitle}{\emph{ICDE}}.
\newblock


\bibitem[Leskovec and Krevl(2014)]%
        {snapnets}
\bibfield{author}{\bibinfo{person}{Jure Leskovec} {and} \bibinfo{person}{Andrej
  Krevl}.} \bibinfo{year}{2014}\natexlab{}.
\newblock \bibinfo{title}{{SNAP Datasets}: {Stanford} Large Network Dataset
  Collection}.
\newblock \bibinfo{howpublished}{\url{http://snap.stanford.edu/data}}.
\newblock


\bibitem[Leskovec et~al\mbox{.}(2007)]%
        {Leskovec2007}
\bibfield{author}{\bibinfo{person}{Jure Leskovec}, \bibinfo{person}{Mary
  McGlohon}, \bibinfo{person}{Christos Faloutsos}, \bibinfo{person}{Natalie~S.
  Glance}, {and} \bibinfo{person}{Matthew Hurst}.}
  \bibinfo{year}{2007}\natexlab{}.
\newblock \showarticletitle{Patterns of Cascading Behavior in Large Blog
  Graphs}. In \bibinfo{booktitle}{\emph{SDM}}.
\newblock


\bibitem[Li et~al\mbox{.}(2013)]%
        {chen2013foe}
\bibfield{author}{\bibinfo{person}{Yanhua Li}, \bibinfo{person}{Wei Chen},
  \bibinfo{person}{Yajun Wang}, {and} \bibinfo{person}{Zhi-Li Zhang}.}
  \bibinfo{year}{2013}\natexlab{}.
\newblock \showarticletitle{Influence Diffusion Dynamics and Influence
  Maximization in Social Networks with Friend and Foe Relationships}. In
  \bibinfo{booktitle}{\emph{WSDM}}.
\newblock


\bibitem[Li et~al\mbox{.}(2018)]%
        {LiFWT18}
\bibfield{author}{\bibinfo{person}{Yuchen Li}, \bibinfo{person}{Ju Fan},
  \bibinfo{person}{Yanhao Wang}, {and} \bibinfo{person}{Kian-Lee Tan}.}
  \bibinfo{year}{2018}\natexlab{}.
\newblock \showarticletitle{Influence maximization on social graphs: A survey}.
\newblock \bibinfo{journal}{\emph{IEEE Transactions on Knowledge and Data
  Engineering}} (\bibinfo{year}{2018}).
\newblock


\bibitem[Lu et~al\mbox{.}(2016)]%
        {LCL2016}
\bibfield{author}{\bibinfo{person}{Wei Lu}, \bibinfo{person}{Wei Chen}, {and}
  \bibinfo{person}{Laks Lakshmanan}.} \bibinfo{year}{2016}\natexlab{}.
\newblock \showarticletitle{From Competition to Complementarity: Comparative
  Influence Diffusion and Maximization}. In \bibinfo{booktitle}{\emph{VLDB}}.
\newblock


\bibitem[Nemhauser et~al\mbox{.}(1978)]%
        {NWF78}
\bibfield{author}{\bibinfo{person}{G.~L. Nemhauser}, \bibinfo{person}{L.~A.
  Wolsey}, {and} \bibinfo{person}{M.~L. Fisher}.}
  \bibinfo{year}{1978}\natexlab{}.
\newblock \bibinfo{booktitle}{\emph{An analysis of the approximations for
  maximizing submodular set functions}}.
\newblock \bibinfo{publisher}{Mathematical Programming}.
\newblock


\bibitem[Nguyen et~al\mbox{.}(2016)]%
        {Nguyen_DSSA_2016}
\bibfield{author}{\bibinfo{person}{Hung~T. Nguyen}, \bibinfo{person}{My~T.
  Thai}, {and} \bibinfo{person}{Thang~N. Dinh}.}
  \bibinfo{year}{2016}\natexlab{}.
\newblock \showarticletitle{Stop-and-Stare: Optimal Sampling Algorithms for
  Viral Marketing in Billion-Scale Networks}. In
  \bibinfo{booktitle}{\emph{SIGMOD}}.
\newblock


\bibitem[Richardson and Domingos(2002)]%
        {richardson02}
\bibfield{author}{\bibinfo{person}{Matthew Richardson} {and}
  \bibinfo{person}{Pedro Domingos}.} \bibinfo{year}{2002}\natexlab{}.
\newblock \showarticletitle{Mining Knowledge-Sharing Sites for Viral
  Marketing}. In \bibinfo{booktitle}{\emph{KDD}}.
\newblock


\bibitem[Simon(1955)]%
        {simon1955class}
\bibfield{author}{\bibinfo{person}{Herbert~A Simon}.}
  \bibinfo{year}{1955}\natexlab{}.
\newblock \showarticletitle{On a class of skew distribution functions}.
\newblock \bibinfo{journal}{\emph{Biometrika}} (\bibinfo{year}{1955}).
\newblock


\bibitem[Sun et~al\mbox{.}(2018)]%
        {sun2018MRIM}
\bibfield{author}{\bibinfo{person}{Lichao Sun}, \bibinfo{person}{Weiran Huang},
  \bibinfo{person}{Philip~S. Yu}, {and} \bibinfo{person}{Wei Chen}.}
  \bibinfo{year}{2018}\natexlab{}.
\newblock \showarticletitle{Multi-Round Influence Maximization}. In
  \bibinfo{booktitle}{\emph{KDD}}.
\newblock


\bibitem[Tang et~al\mbox{.}(2009)]%
        {TangSWY09}
\bibfield{author}{\bibinfo{person}{Jie Tang}, \bibinfo{person}{Jimeng Sun},
  \bibinfo{person}{Chi Wang}, {and} \bibinfo{person}{Zi Yang}.}
  \bibinfo{year}{2009}\natexlab{}.
\newblock \showarticletitle{Social Influence Analysis in Large-Scale Networks}.
  In \bibinfo{booktitle}{\emph{KDD}}.
\newblock


\bibitem[Tang et~al\mbox{.}(2018)]%
        {Tang_OPIM_2018}
\bibfield{author}{\bibinfo{person}{Jing Tang}, \bibinfo{person}{Xueyan Tang},
  \bibinfo{person}{Xiaokui Xiao}, {and} \bibinfo{person}{Junsong Yuan}.}
  \bibinfo{year}{2018}\natexlab{}.
\newblock \showarticletitle{Online Processing Algorithms for Influence
  Maximization}. In \bibinfo{booktitle}{\emph{SIGMOD}}.
\newblock


\bibitem[Tang et~al\mbox{.}(2015)]%
        {tang15}
\bibfield{author}{\bibinfo{person}{Youze Tang}, \bibinfo{person}{Yanchen Shi},
  {and} \bibinfo{person}{Xiaokui Xiao}.} \bibinfo{year}{2015}\natexlab{}.
\newblock \showarticletitle{Influence Maximization in Near-Linear Time: A
  Martingale Approach}. In \bibinfo{booktitle}{\emph{SIGMOD}}.
\newblock


\bibitem[Tang et~al\mbox{.}(2014)]%
        {tang14}
\bibfield{author}{\bibinfo{person}{Youze Tang}, \bibinfo{person}{Xiaokui Xiao},
  {and} \bibinfo{person}{Yanchen Shi}.} \bibinfo{year}{2014}\natexlab{}.
\newblock \showarticletitle{Influence Maximization: Near-Optimal Time
  Complexity Meets Practical Efficiency}. In
  \bibinfo{booktitle}{\emph{SIGMOD}}.
\newblock


\bibitem[Wang et~al\mbox{.}(2012)]%
        {WCW12}
\bibfield{author}{\bibinfo{person}{Chi Wang}, \bibinfo{person}{Wei Chen}, {and}
  \bibinfo{person}{Yajun Wang}.} \bibinfo{year}{2012}\natexlab{}.
\newblock \showarticletitle{Scalable influence maximization for independent
  cascade model in large-scale social networks}.
\newblock \bibinfo{journal}{\emph{Data Mining and Knowledge Discovery}}
  (\bibinfo{year}{2012}).
\newblock


\bibitem[Wang et~al\mbox{.}(2010)]%
        {wang2010community}
\bibfield{author}{\bibinfo{person}{Yu Wang}, \bibinfo{person}{Gao Cong},
  \bibinfo{person}{Guojie Song}, {and} \bibinfo{person}{Kunqing Xie}.}
  \bibinfo{year}{2010}\natexlab{}.
\newblock \showarticletitle{Community-Based Greedy Algorithm for Mining Top-K
  Influential Nodes in Mobile Social Networks}. In
  \bibinfo{booktitle}{\emph{KDD}}.
\newblock


\bibitem[Yule(1925)]%
        {yule1925ii}
\bibfield{author}{\bibinfo{person}{George~Udny Yule}.}
  \bibinfo{year}{1925}\natexlab{}.
\newblock \showarticletitle{II.鈥擜 mathematical theory of evolution, based on
  the conclusions of Dr. JC Willis, FR S}.
\newblock \bibinfo{journal}{\emph{Philosophical transactions of the Royal
  Society of London. Series B, containing papers of a biological character}}
  (\bibinfo{year}{1925}).
\newblock


\end{thebibliography}

\clearpage

\appendix

\section{Notations}
Table.\ref{tab:notations} shows all terms' notations in this paper.

    \begin{table*}
\caption{Notations}
\label{tab:notations}
\begin{tabular}{|c|c|}
\hline
Notation & Description       \\
\hline
$G=(V,E)$ & a social network $G$ with a node set $V$ and an edge set $E$ \\
\hline
$N,M$  &  the numbers of nodes and edges in $G$, respectively         \\
\hline
$d_t^n$ & the popularity measure of Novice item at the end of round $t$   \\
\hline
$d_t^p$ & the popularity measure of Popular item at the end of round $t$   \\
\hline
$k$ & number of seeds to be selected \\
\hline
$\tau$ & time step index, each time step is one step in one round of IC model propagation  \\
\hline
$t$ & round index, each round is a promotional round  \\
\hline
$T$ & Total rounds  \\
\hline
$S_t$&seed set of round $t$ \\
\hline
$\cS = \bigcup_{t=1}^T S_t \times \{t\}$& pair set, a pair $(u,t)$ is a node $u$ at round $t$ \\
\hline
$\cS^*$& optimal pair set \\
\hline
$r_T(\cS)$& the popularity ratio at the end of round $T$ \\
\hline
$r_0$& the initial popularity ratio \\
\hline
$p(u,v)$& the probability of node $u$ active $v$ \\
\hline
$z$& total natural growth of popularity measure \\
\hline
$\sigma(S)$ & the influence spread of seed set $S$ \\
\hline
$w=(w_1,w_2...w_T)$ & the weight vector, $w_t$ is the weight of round $t$ \\
\hline
$\rho_T(\cS)$ & weighted overlapping influence spread \\
\hline
$\rho^{(t)}(S_t)$ & weighted overlapping influence spread of round $t$\\
\hline
$R, R(v)$ & one RR set, $R(v) $ is the RR set rooted at node $v$  \\
\hline
$R^{(t)}, R^{(t)}(v)$ & one PW-RR set, $R^{(t)}= R\times \{t\}$, $R^{(t)}(v) $ is the RR set rooted at pair $(v,t)$  \\
\hline
$\cR$ & the set of pair-wised RR set  \\

 \hline
$\theta$ & the number of PW-RR sets that need to be generated \\
\hline
$l$ & Error probability parameter \\
\hline
$\varepsilon$ & Approximate ratio parameter \\
\hline
$b,\cQ,i,j,x_i,\alpha,\beta,OPT,\delta,LB,Pr,\omega,\Omega$ & parameters in derivation \\
\hline

\end{tabular}
\end{table*}

\section{Proof of Monotonicity and Non-submodularity of Popularity Ratio Function (lemma~\ref{lemma:submodular of popularity ratio})}
\begin{proof}
It is proved that the expected influence spread set function $\sigma$ is monotone and submodular, we don't show the proof process here. So we only need to prove that the popularity ratio function is monotone with respect to the function $\sigma(S_i),0<i<T$. We can prove the monotone of the popularity ratio function. For simplicity, we denote $d_{t-1}^n+d_{t-1}^p+z$ as $d_t$, denote $\sigma(S_i)$ as $x_i,$ $0<i<T$.
\begin{displaymath}
    \begin{aligned}
    \frac{\partial r_T}{\partial x_i} 
    &= \left(r_0+1\right)\prod_{t=1}^{i-1}\left(1+\frac{x_t}{d_t}\right)\Bigg[\frac{1}{d_i}\prod_{t=i+1}^{T}\left(1+\frac{x_t}{d_t}\right)\\
    &+\sum_{t=i+1}^{T}\frac{-x_t}{d_t^2}\prod_{s=i,s\neq t}^{T}\left(1+\frac{x_s}{d_s}\right)\Bigg] \\
    &=r_T\cdot\Bigg[\frac{z}{d_{i+1}\left(d_i+x_i\right)}+\cdots+\frac{z}{d_T\left(d_{T-1}+\sigma(S_{T-1})\right)}\\
    &+\frac{1}{d_T+x_T}\Bigg]>0
    \end{aligned}
\end{displaymath}
For any $0<i<T, \frac{\partial r_T}{\partial x_i}>0$, the popularity ratio function is monotone. 

However the popularity ratio function does not satisfy submodular, and we will illustrate this property with a counter example below.

Consider the simplest case that social network has only three nodes $u,v,w$. No edges between nodes, the initial popularity measure of Novice item is $d_0^n=1$ and the initial popularity measure of Popular item is $d_0^p=2$, the increment of popularity measure of each round $z=1$. In this case, the original 
\begin{displaymath}
    r_T(\cS)= \left(r_0+1\right)\prod_{t=1}^T \left(1+\frac{\sigma(S_t)}{d_0^n+d_0^p+z\cdot t+\sum_{i=1}^{t-1} \sigma(S_i)}\right)-1
\end{displaymath}
where $T$ is at most $3$,  $r_0=\frac{d_0^n}{d_0^p}=\frac{1}{2}$.
\begin{displaymath}
\begin{aligned}
        &r_T(\cS)\\ 
        &=\frac{3}{2}*(1+\frac{\sigma(S_1)}{d_0^n+d_0^p+z})*(1+\frac{\sigma(S_2)}{d_1^n+d_1^p+z})*(1+\frac{\sigma(S_3)}{d_2^n+d_2^p+z})-1
\end{aligned}
\end{displaymath}
where $d_0^n+d_0^p+z=4$, $d_1^n+d_1^p+z=5+\sigma(S_1)$, $d_1^n+d_1^p+z=6+\sigma(S_1)+\sigma(S_2)$.
\begin{displaymath}
\begin{aligned}
    &r_T(\cS) \\
    &=\frac{3}{2}*(1+\frac{\sigma(S_1)}{4})*(1+\frac{\sigma(S_2)}{5+\sigma(S_1)})*(1+\frac{\sigma(S_3)}{6+\sigma(S_1)+\sigma(S_2)})-1
\end{aligned}
\end{displaymath}
Now the two pair set $\cS\subset \cQ$, $\cS=\{(u,1)\}$, $\cQ=\{(u,1), (v,1)\}$, and a pair $b=(w,2)$, Clearly $b \notin \cS$.
\begin{displaymath}
\begin{aligned}
    &r_T(\cS) = 0.875,  r_T(\cS\cup \{b\})= 1.1875\\ 
    &r_T(\cQ)=1.25, r_T(\cQ\cup \{b\})=1.5714
\end{aligned}
\end{displaymath}
\begin{displaymath}
    r_T(\cQ\cup \{b\})-r_T(\cQ)=0.3214, r_T(\cS\cup \{b\})-r_T(\cS)=0.3125
\end{displaymath}
\begin{displaymath}
    r_T(\cQ\cup \{b\})-r_T(\cQ)>r_T(\cS\cup \{b\})-r_T(\cS)
\end{displaymath}

In this case, the marginal value of $\cQ$ is larger than $\cS$. So the set popularity ratio function is not submodular.
\end{proof}

\section{Simplification Process from Popularity Ratio Function to Round Weighted Influence Function (Section \ref{sec:simplify}}
The first step: expanding the multiplication series of Eq.\eqref{eq:ratio plus one(OINS)} and only keeping the first-order terms;
\begin{displaymath}
\begin{aligned}
    &\prod_{t=1}^T \left(1+\frac{\sigma(S_t)}{d_0^n+d_0^p+z\cdot t+\sum_{i=1}^{t-1} \sigma(S_i)}\right)\\
    &=\frac{\sigma(S_1)}{d_0^n+d_0^p+z} + \frac{\sigma(S_2)}{d_0^n+d_0^p+2z+\sigma(S_1)} \\
    &+ \cdots + \frac{\sigma(S_T)}{d_0^n+d_0^p+T\cdot z + \sum_{i=0}^{T-1} \sigma(S_i)}\\
\end{aligned}
\end{displaymath}
second step: removing the $\sigma(S_1),\ldots, \sigma(S_{T-1})$ in the denominator of each term left after step (a). 
\begin{displaymath}
\begin{aligned}
&=\frac{\sigma(S_1)}{d_0^n+d_0^p+z} + \frac{\sigma(S_2)}{d_0^n+d_0^p+2z} + \cdots + \frac{\sigma(S_T)}{d_0^n+d_0^p+T\cdot z}
\end{aligned}
\end{displaymath}
Combining the above simplified process, it can be noted that
\begin{displaymath}
    \rho_T(\cS) = \frac{\sigma(S_1)}{d_0^n+d_0^p+z} + \frac{\sigma(S_2)}{d_0^n+d_0^p+2z} + \cdots + \frac{\sigma(S_T)}{d_0^n+d_0^p+T\cdot z}
\end{displaymath}
is our weighted overlapping influence function, where $d_0^n$, $d_0^p$, $z$ are our predefined parameters. Thus we can denote the weighted overlapping influence function as
\begin{equation}
    \rho_T(\cS)=\sum_{t=1}^{T} {w_{t}\cdot\sigma(S_{t})}
\end{equation}

\section{PW-RR generation process}
RR set $R$ is generated by independently reverse simulating the propagation from $v$ in round $t$. A (random) pair-wise RR set $(R,t)$ is a RR set $R$ rooted at a node picked uniformly at random from $V$, and $t$ is picked uniformly at random from $[T]$.

\begin{algorithm}
\label{alg:PW-RR}
\caption{PW-RR generation process.}
\leftline{\textbf{Input:}  directed graph $G=(V,E)$, IC model, Max round $T$} 
\leftline{Number of PW-RR set $\theta$}
\leftline{\textbf{Output:} the set of PW-RR set $\cR$}
\begin{algorithmic}[1]
\STATE $\cR = \emptyset$
\FOR{$0 < \theta$}
\STATE $\theta = \theta - 1$
\STATE Generate an RR set $R$ for a random node $v \in V$
\STATE choose a round $t$ uniformly at random from $[T]$
\STATE put the RR set $R$ and round $t$ together as $R^{(t)}$
\STATE $\cR = \cR \cup \{R^{(t)}\}$
\ENDFOR
\STATE return $\cR$
\end{algorithmic}
\end{algorithm}

\section{Proof of Monotonicity and Non-submodularity of Popularity Ratio Function (lemma~\ref{lem:rhoproperty}) }
\begin{proof}
	
For every $t\in [T]$ and every set $\cS$ of pairs in $V\times [T]$, define $\rho^{(t)}(\cS) = \sigma(S_t)$.
Using the fact that the influence spread function $\sigma(S)$ is monotone and submodular with respect to $S$, 
	we want to show that $\rho^{(t)}(\cS)$ is monotone and submodular with respect to $\cS$. 
In fact, for every $\cS \subseteq \cQ \subseteq V\times [T]$, we know that $S_t \subseteq Q_t$, and therefore
	$\rho^{(t)}(\cS) = \sigma(S_t) \le \sigma(Q_t) = \rho^{(t)}(\cQ)$, and thus the monotonicity holds.

Now, suppose that $\cS \subseteq \cQ \subseteq V\times [T]$ and $b=(v,j) \in V\times [T] \setminus \cQ$.
If $j \ne t$, then $\cS \cup \{b\}$ and $\cS$ has the same node set for round $t$, which means
	$\rho^{(t)}(\cS \cup \{b\}) - \rho^{(t)}(\cS) = 0$.
Similarly, $\rho^{(t)}(\cQ \cup \{b\}) - \rho^{(t)}(\cQ) = 0$.
Thus, $\rho^{(t)}(\cQ \cup \{b\}) - \rho^{(t)}(\cQ) \le \rho^{(t)}(\cS \cup \{b\}) - \rho^{(t)}(\cS)$.
If $j=t$, then we have $S_t \subseteq Q_t$ and $v \in V \setminus Q_t$.
By the submodularity of $\sigma$, we have
	$\rho^{(t)}(\cQ \cup \{b\}) - \rho^{(t)}(\cQ) = \sigma(Q_t \cup \{v\}) - \sigma(Q_t)
	\le \sigma(S_t \cup \{v\}) - \sigma(S_t) = \rho^{(t)}(\cS \cup \{b\}) - \rho^{(t)}(\cS)$.
Therefore submodularity also holds.

Finally, since $\rho^{(t)}(\cS)$ is monotone and submodular with respect to $\cS$ for every $t$, 
	by the well known fact that the nonnegative weighted summation of monotone submodular functions is
	still monotone and submodular,
	we know that $\rho_T(\cS)=\sum_{t=1}^{T} w_t\sigma(S_{t}) = \sum_{t=1}^{T} w_t\rho^{(t)}(\cS)$ is also monotone and submodular with respect to $\cS$. 
\end{proof}

\section{Proof of property of PW-RR (lemma \ref{lemma:property of PW-RR})}
\begin{proof}
The randomness of $Y(\cS)$ is from two aspect: (1) the root of a PW-RR set is uniformly random choose, (2) the round $t$ of the PW-RR set is uniformly random choose.

\begin{displaymath}
\begin{aligned}
    &\mathbb{E}\left[Y(\cS) \right]\\
    & =\frac{1}{T} \sum_{t=1}^{T} w_t \cdot \mathbb{E}\left(\mathbb{I}\left\{\cS \cap \cR \neq \emptyset\right\}\right) \\
    &=\frac{1}{T} \sum_{t=1}^{T} w_t \cdot \operatorname{Pr}\left\{S_t \cap R^{(t)} \neq \emptyset\right\} \\
    &=\frac{1}{T} \sum_{t=1}^{T} w_t \cdot \frac{1}{N} \sum_{v \in V} \operatorname{Pr}\left\{S_t \cap R^{(t)}(v) \neq \emptyset\right\} \\
    &=\frac{1}{T} \sum_{t=1}^{T} w_t \cdot \frac{1}{N} \sum_{v \in V} a p\left(S_t, v\right) \\
    &=\frac{1}{T} \sum_{t=1}^{T} w_t \cdot \frac{1}{N} \sigma\left(S_t\right)\\
\end{aligned}
\end{displaymath}
For any seed set $S_t\in V$, any node $v\in V$, the probability that the seed set $S_t$ activates node $v$ with probability $ap(S_t, v)$. $ap(S_t, v)$ is the probability that $S_t$ have an intersaction with a random RR set $R(v)$ rooted from node $v$. i.e. $ap(S,v) = \operatorname{Pr}\{S\cap R(v) \neq \emptyset\}$

\end{proof}





\section{Correctness of PRM-IMM Algorithm (Theorem \ref{theorem:main result})}

We first give a general conclusion (Theorem \ref{theorem:condition of approximate guarantee}) to show how the greedy solution obtained by the PRM-NodeSelection approaches the optimal solution of the weighted overlapping influence maximization problem when the $\hat{\rho}_T\left(\cS, \cR\right)$ itself satisfies monotone submodular (It is easy to proof that $\hat{\rho}_T\left(\cS, \cR\right)$ is submodular and nondecreasing with respect to $\cS$). 

We denote the random estimation of $\rho_T(\cS)$ as $\hat{\rho}_T(\cS,\omega)$, where $\omega\in\Omega$ is a sample in random space $\Omega$. $\cS^{*}$ is the optimal solution of $\rho_T(\cS)$, $\OPT = \rho_T(\cS^*)$. $\hat{\cS^g}(\omega)$ is the greedy result of $\hat{\rho}_T(\cdot,\omega)$. For $\varepsilon > 0$, we say a solution $\cS$ is bad, if $\rho_T(\cS) < (1/2-\varepsilon)\cdot \OPT$.
\begin{theorem}
\label{theorem:condition of approximate guarantee}
for any $\varepsilon > 0$, $\varepsilon_{1} \in\left(0, 2\varepsilon\right)$, $\delta_1, \delta_2 >0$, if:
\begin{itemize}
    \item [(a)] $\underset{\omega \sim \Omega}{\operatorname{Pr}}\left\{\hat{\rho}_T\left(\cS^*, \omega\right) \geq\left(1-\varepsilon_{1}\right) \cdot \OPT \right\} \geq 1-\delta_{1}$ 
    \item [(b)] for any bad $\cS$ , \\$\underset{\omega \sim \Omega}{\operatorname{Pr}}\left\{\hat{\rho}_T\left(\cS, \omega\right) \geq\frac{1}{2}\left(1-\varepsilon_{1}\right) \cdot \OPT\right\}\leq\frac{\delta_{2}}{T^k\cdot\binom{N}{k}}$
    \item [(c)] for all $\omega \sim \Omega, \hat{\rho}_T(\cS,\omega)$is monotone and submodular with respect to $\cS$.
\end{itemize}
So that, $\underset{\omega \sim \Omega}{\operatorname{Pr}}\left\{\rho_T\left(\hat{\cS}^g(\omega)\right) \geq\left(\frac{1}{2}-\varepsilon\right) \cdot \OPT\right\}\geq1-\delta_{1}-\delta_{2}$

\end{theorem}
\begin{proof}
Because $\hat{\rho}_T(\cS,\omega)$is monotone and submodular with respect to $\cS$, due to the property of partition matroid, we know the greedy solution returns a 1/2-approximate solution.
\begin{displaymath}
    \hat{\rho}_T\left(\hat{\cS}^{g}(\omega), \omega\right) \geq\frac{1}{2} \hat{\rho}_T\left(\cS^{*}, \omega\right)
\end{displaymath}
With (a), we have at least $1-\delta_{1}$ probability that:
\begin{displaymath}
    \hat{\rho}_T\left(\hat{\cS}^{g}(\omega), \omega\right) \geq\frac{1}{2}\left(1-\varepsilon_{1}\right) \cdot O P T
\end{displaymath}
In $T$ round, the number of the $k$ size seed set is at most $T^k\cdot\binom{N}{k}$. The probability that every bad $\cS$ satisfy (b) is less than $\frac{\delta_{2}}{T^k\cdot\binom{N}{k}}$, so the probability that existing a $\cS$ to make $\rho_T\left(\cS, \omega\right) \geq\frac{1}{2}\left(1-\varepsilon_{1}\right) \cdot \OPT$ is at most $\delta_{2}$.

So $\underset{\omega \sim \Omega}{\operatorname{Pr}}\left\{\rho_T\left(\hat{\cS}^g(\omega)\right) \geq\left(\frac{1}{2}-\varepsilon\right) \cdot \OPT\right\}\geq1-\delta_{1}-\delta_{2}$
\end{proof}
We will use the concentration inequality to find out how much $\theta$ is sufficient to satisfy the conditions (a) and (b) in Theorem \ref{theorem:condition of approximate guarantee}. For all subsequences of length $\theta, \cR[\theta]$ in the probability space $\Omega$, each PW-RR set is also independent of each other, so we can use Chernoff bounds of independent sequences to analyze, which is more simple and intuitive\cite{chen18}.
\begin{theorem}
\label{theorem:condition of approximate guarantee 2}
For any $\varepsilon > 0$, $\varepsilon_{1} \in \left(0,2\varepsilon\right)$, $\delta_{1},\delta_{2}>0$:
\begin{displaymath}
    \theta^{(1)}=\frac{2 w_1 N \cdot T \cdot \ln \frac{1}{\delta_{1}}}{\OPT \cdot \varepsilon_{1}^{2}}, 
    \theta^{(2)}=\frac{w_1 N\cdot T \cdot \ln \left(\frac{T^{k} \cdot\binom{N}{k}}{\delta_{2}}\right)}{O P T \cdot\left(\varepsilon-\frac{1}{2}\varepsilon_{1}\right)^{2}}
\end{displaymath}
For any fixed $\theta>\theta^{(1)},\underset{\omega \sim \Omega}{\operatorname{Pr}}\left\{\hat{\rho}_T\left(\cS^*,\omega\right)\geq\left(1-\varepsilon_1\right)\cdot \OPT\right\}\geq 1-\delta_1$ \newline
For any fixed $\theta>\theta^{(2)}$, any bad $\cS$,
\begin{displaymath}
    \underset{\omega \sim \Omega}{\operatorname{Pr}}\left\{\hat{\rho}_T(\cS, \omega) \geq\frac{1}{2}\left(1-\varepsilon_{1}\right) \cdot O P T\right\} \leq \frac{\delta_{2}}{T^{k} \cdot \binom{N}{k}}
\end{displaymath}
\end{theorem}
\begin{proof}
When $\theta>\theta^{(1)}$
Notice that:
\begin{displaymath}
\begin{aligned}
    &\hat{\rho}_T(\cS, \cR)=\frac{N \cdot T}{\theta} \sum_{j=1}^{\theta} Y_j^{\cR}(\cS)\\
    &\underset{\cR_{0} \sim \Omega}{\operatorname{Pr}}\left\{\hat{\rho}_T\left(\cS^*,\cR_{0}\right)<\left(1-\varepsilon_1\right)\cdot \OPT\right\}\\
    &=\underset{\cR_{0} \sim \Omega}{\operatorname{Pr}}\left\{\frac{N\cdot T}{\theta} \cdot \sum_{j=1}^{\theta} Y_{j}^{\cR_0}\left(\cS^{*}\right)<\left(1-\varepsilon_{1}\right) \cdot \OPT\right\} \\
\end{aligned}
\end{displaymath}
\begin{displaymath}
\begin{aligned}
    &=\underset{\cR_{0} \sim \Omega}{\operatorname{Pr}}\left\{\sum_{j=1}^{\theta} Y_{j}^{\cR_0}\left(\cS^{*}\right)<\frac{\theta}{N\cdot T} \cdot\left(1-\varepsilon_{1}\right) \cdot \OPT\right\}\\
    &=\underset{\cR_{0} \sim \Omega}{\operatorname{Pr}}\left\{\sum_{j=1}^{\theta} Y_{j}^{\cR_0}\left(\cS^{*}\right)-\theta \cdot \frac{\rho_T\left(\cS^{*}\right)}{N\cdot T}<\frac{\theta\left(1-\varepsilon_{1}\right)\OPT}{N\cdot T} -\theta\frac{\rho_T\left(\cS^{*}\right)}{N\cdot T}\right\} \\
    &=\underset{\cR_{0} \sim \Omega}{\operatorname{Pr}}\left\{\sum_{j=1}^{\theta} Y_{j}^{\cR_0}\left(\cS^{*}\right)-\theta \cdot \frac{\rho_T\left(\cS^{*}\right)}{N\cdot T}<-\varepsilon_{1} \cdot\left(\theta \cdot \frac{\rho_T\left(\cS^{*}\right)}{N\cdot T}\right)\right\}\\
    &=\underset{\cR_{0} \sim \Omega}{\operatorname{Pr}}\left\{\sum_{j=1}^{\theta} \frac{Y_{j}^{\cR_0}\left(\cS^{*}\right)}{w_1} -\theta \cdot \frac{\rho_T\left(\cS^{*}\right)}{w_1 N\cdot T}<-\varepsilon_{1} \cdot\left(\theta \cdot \frac{\rho_T\left(\cS^{*}\right)}{w_1 N\cdot T}\right)\right\}\\
    &\leq \exp \left(-\frac{\varepsilon_{1}^{2}}{2} \theta \cdot \frac{\rho_T\left(\cS^{*}\right)}{w_1 N\cdot T}\right)\\
    &\leq \exp \left(-\frac{\varepsilon_{1}^{2}}{2} \cdot \frac{2 w_1 N\cdot T \cdot \ln \frac{1}{\delta_{1}}}{\OPT \cdot \varepsilon_{1}^{2}} \cdot \frac{\rho_T\left(\cS^{*}\right)}{N\cdot T}\right)=\delta_{1}
\end{aligned}
\end{displaymath}
When $\theta>\theta^{(2)}$, set $\varepsilon_2=\varepsilon-\frac{1}{2}\varepsilon_1$, $\rho_T(\cS)<\left(\frac{1}{2}-\varepsilon\right)\cdot\OPT$
\begin{displaymath}
\begin{aligned}
    &\underset{\cR_{0}\sim\Omega}{\operatorname{Pr}}\left\{\hat{\rho}_{T\theta}\left(\cS, \cR_{0}\right) \geq\frac{1}{2}\left(1-\varepsilon_{1}\right) \cdot \OPT\right\} \\
    &=\underset{\cR_{0} \sim \Omega}{\operatorname{Pr}}\left\{\frac{N\cdot T}{\theta} \cdot \sum_{j=1}^{\theta} Y_{j}^{\cR_0}(\cS) \geq\frac{1}{2}\left(1-\varepsilon_{1}\right) \cdot \OPT\right\}\\
    &=\underset{\cR_{0}\sim\Omega}{\operatorname{Pr}} \left\{\sum_{j=1}^{\theta} Y_{j}^{\cR_0}(\cS)-\theta \cdot \frac{\rho_T(\cS)}{N\cdot T} \geq \frac{\theta}{N\cdot T}\left[\frac{1}{2}\left(1-\varepsilon_{1}\right) \cdot \OPT-\rho_T(\cS)\right]\right\}\\ 
    &/* \rho_T(\cS)<\left(\frac{1}{2}-\varepsilon\right)\cdot\OPT*/\\
    &\leq \underset{\cR_{0}\sim\Omega}{\operatorname{Pr}}\left\{\sum_{j=1}^{\theta} Y_{j}^{\cR_0}(\cS)-\theta \cdot \frac{\rho_T(\cS)}{N\cdot T} \geq \frac{\theta}{N\cdot T} \cdot \varepsilon_{2} \cdot \OPT\right\}\\
    &=\underset{\cR_{0} \sim \Omega}{\operatorname{Pr}}\left\{\sum_{j=1}^{\theta} \frac{Y_{j}^{\cR_0}(\cS)}{w_1}-\theta \cdot \frac{\rho_T(\cS)}{w_1 N\cdot T} \geq\left(\varepsilon_{2} \cdot \frac{\OPT}{\rho_T(\cS)}\right) \cdot \theta \cdot \frac{\rho_T(\cS)}{w_1 N\cdot T}\right\} \\
    &\leq \exp \left(-\frac{\left(\varepsilon_{2} \cdot \frac{\OPT}{\rho_T(\cS)}\right)^{2}}{2+\frac{2}{3}\left(\varepsilon_{2} \cdot \frac{\OPT}{\rho_T(\cS)}\right)} \cdot \theta \cdot \frac{\rho_T(\cS)}{w_1 N\cdot T}\right) \\
    &\leq \exp \left(-\frac{\varepsilon_{2}^{2} \cdot \OPT^{2}}{2 \rho_T(\cS)+\frac{2}{3} \varepsilon_{2} \cdot \OPT} \cdot \theta \cdot \frac{1}{w_1 N\cdot T}\right)
\end{aligned}
\end{displaymath}

\begin{displaymath}
\begin{aligned}
    &\leq \exp \left(-\frac{\varepsilon_{2}^{2} \cdot \OPT^{2}}{2\left(\frac{1}{2}-\varepsilon\right) \cdot \OPT+\frac{2}{3}\left(\varepsilon-\frac{1}{2} \varepsilon_{1}\right) \cdot \OPT} \cdot \theta \cdot \frac{1}{w_1 N\cdot T}\right)\\
    &\leq \exp \left(-\frac{\left(\varepsilon-\frac{1}{2} \varepsilon_{1}\right)^{2} \cdot \OPT^{2}}{\OPT} \cdot \frac{w_1 N\cdot T \cdot \ln \left(\frac{T^{k} \cdot\binom{N}{k}}{\delta_{2}}\right)}{\OPT \cdot\left(\varepsilon-\frac{1}{2} \varepsilon_{1}\right)^{2}} \cdot \frac{1}{w_1 N\cdot T}\right)\\
    &=\frac{T^{k} \cdot\binom{N}{k}}{\delta_{2}}
\end{aligned}
\end{displaymath}

\end{proof}

Now we discuss the setting of parameters $\varepsilon_1, \delta_1, \delta_2$. The Settings of these parameters are not unique, and the method we describe below follows the settings in the original IMM algorithm. According to Theorem \ref{theorem:condition of approximate guarantee 2}, assuming that $\OPT$ is known, the target of these parameters is to make the output of greedy solution $\hat{\cS^g}$ is the $1/2 - \varepsilon$ approximation of optimal solution with probability at least $1-1/(2N^l)$. The high probability of $1-1/(2N^l)$ was achieved because, in the next step, we would use the same high probability of $1-1/(2N^l)$to obtain a better lower-bound estimate of $\OPT$. Thus, the correctness of the overall algorithm would be guaranteed for an assignment with a high probability of $1-1/(N^l)$. The following corollary give the setting of the parameters.
\begin{corollary}\label{Corollary:1}
    Set $\delta_1=\delta_2=\frac{1}{4n^l}, \varepsilon_1=\varepsilon\cdot \frac{\alpha}{\frac{1}{2}\alpha+\beta}$\newline
\begin{displaymath}
    \alpha=\sqrt{l \ln N+\ln 4}, \beta=\sqrt{\frac{1}{2} \cdot\left(\ln \binom{N}{k} + l \ln N+\ln 4+k \ln T\right)}
\end{displaymath}
    For any fixed $\theta>\frac{2N\cdot T\cdot\left[\frac{1}{2}\alpha+\beta\right]^2}{\varepsilon^2\cdot \OPT}$, if the input of PRM-NodeSelection is $\cR_0[\theta],\cR_0\sim\Omega$,the probability that PRM-NodeSelection-OINS's output $\hat{\cS}^g(\cR_0[\theta])$ is the $(1/2-\varepsilon)$ approximation of the optimal solution is at least $1-\frac{1}{2N^l}$.
\end{corollary}



\begin{theorem}\label{theorem:prob of LB<OPT}
The probability of $LB \leq\OPT$ is at least $1-\frac{1}{2N^l}$, which means that the probability of $\tilde{\theta}\geq\frac{2nt\cdot\left[\frac{1}{2}\cdot\alpha+\beta\right]^2}{\varepsilon^2\cdot\OPT}$ is at least $1-\frac{1}{2N^l}$.
\end{theorem}

\subsection{Proof of Theorem \ref{theorem:prob of LB<OPT}}
\begin{theorem}\label{theorem:LB<OPT}
For any $i=1,2,\cdots,\left \lfloor \log_{2}N \right \rfloor-1 $,
\begin{itemize}
    \item [(1)] if $x_{i}=\frac{\sum_{1}^{k} w_t \cdot N}{2^{i}}>\OPT$, the probability of $\hat{\rho}_{T\theta_{i}} \left(\cS_{i},\cR_{0}[\theta_i]\right)\geq\left(1+\varepsilon^{'}\right)\cdot x_{i}$  is at most $\frac{1}{2N^{l}\log_{2}N}$.
    \item [(2)] if $x_{i}=\frac{\sum_{1}^{k} w_{t} \cdot N}{2^{i}}\leq\OPT$, the probability of $\hat{\rho}_{T\theta_{i}} \left(\cS_{i},\cR_{0}[\theta_i]\right)\geq\left(1+\varepsilon^{'}\right)\cdot\OPT$ is at most $\frac{1}{2N^{l}\log_{2}N}$.
\end{itemize}
\end{theorem}
\begin{proof}
For any $k$ size seed set $\cS$. 

$\hat{\rho}_{T\theta_i} \left(\cS,\cR_0[\theta_i]\right) = \frac{N\cdot T \cdot \sum_{j=1}^{\theta_i} Y_{j}^{\cR_{0}[\theta_i]}(\cS)}{\theta_i}$
\begin{displaymath}
\begin{aligned}
    &\underset{\cR_{0} \sim \Omega}{\operatorname{Pr}} \left\{\hat{\rho}_{T\theta_{i}}\left(\cS, \cR_{0}\right) \geq\left(1+\varepsilon^{\prime}\right) \cdot x_{i}\right\}\\
    &=\underset{\cR_{0} \sim \Omega}{\operatorname{Pr}}\left\{\frac{N\cdot T}{\theta_{i}} \cdot \sum_{j=1}^{\theta_{i}} Y_{j}^{\cR_{0}\left[\theta_{i}\right]}(\cS) \geq\left(1+\varepsilon^{\prime}\right) \cdot x_{i}\right\}\\
    &=\underset{\cR_{0} \sim \Omega}{\operatorname{Pr}}\left\{\sum_{j=1}^{\theta_{i}} Y_{j}^{\cR_{0}\left[\theta_{i}\right]}(\cS)-\frac{\theta_{i} \cdot \rho_T(\cS)}{N\cdot T} \geq \frac{\theta_{i}\left(1+\varepsilon^{\prime}\right) x_{i}}{N\cdot T}-\frac{\theta_{i} \cdot \rho_T(\cS)}{N\cdot T}\right\}\\
    &/ * \text { because } x_{i}>\OPT \geq \rho_T(\cS) * /\\
    &\leq \underset{\cR_{0} \sim \Omega}{\operatorname{Pr}}\left\{ \sum_{j=1}^{\theta_{i}} Y_{j}^{\cR_{0}\left[\theta_{i}\right]}(\cS)-\frac{\theta_{i} \cdot \rho_T(\cS)}{N\cdot T} \geq \frac{\theta_{i}}{N\cdot T} \varepsilon^{\prime} \cdot x_{i}\right\} \quad \\
    &=\underset{\cR_{0} \sim \Omega}{\operatorname{Pr}}\left\{\sum_{j=1}^{\theta_{i}} Y_{j}^{\cR_{0}\left[\theta_{i}\right]}(\cS)-\frac{\theta_{i} \cdot \rho_T(\cS)}{N\cdot T} \geq \frac{\varepsilon^{\prime} \cdot x_{i}}{\rho_T(\cS)} \cdot \frac{\theta_{i} \cdot \rho_T(\cS)}{N\cdot T}\right\}\\
    &=\underset{\cR_{0} \sim \Omega}{\operatorname{Pr}}\left\{\sum_{j=1}^{\theta_{i}} \frac{Y_{j}^{\cR_{0}\left[\theta_{i}\right]}(\cS)}{w_1}-\frac{\theta_{i} \cdot \rho_T(\cS)}{w_1 N\cdot T} \geq \frac{\varepsilon^{\prime} \cdot x_{i}}{\rho_T(\cS)} \cdot \frac{\theta_{i} \cdot \rho_T(\cS)}{w_1 N\cdot T}\right\}\\
    &\leq \exp \left(-\frac{\left(\frac{\varepsilon^{\prime} \cdot x_{i}}{\rho_T(\cS)}\right)^{2}}{2+\frac{2}{3}\left(\frac{\varepsilon^{\prime} \cdot x_{i}}{\rho_T(\cS)}\right)} \cdot \frac{\theta_{i} \cdot \rho_T(\cS)}{w_1 N\cdot T}\right) \quad / * \text { chernoff bound } * /\\
    &=\exp \left(-\frac{\left(\varepsilon^{\prime} \cdot x_{i}\right)^{2}}{2 \rho_T(\cS)+\frac{2}{3}\left(\varepsilon^{\prime} \cdot x_{i}\right)} \cdot \frac{\theta_{i}}{w_1 N\cdot T}\right)\\
    &\leq \exp \left(-\frac{\varepsilon^{\prime 2} \cdot x_{i}}{2+\frac{2}{3} \varepsilon^{\prime}} \cdot \frac{\theta_{i}}{w_1 N\cdot T}\right)\\
\end{aligned}
\end{displaymath}

\begin{displaymath}
\begin{aligned}
    &\leq \exp \Bigg(-\frac{\varepsilon^{\prime 2} \cdot x_{i}}{2+\frac{2}{3} \varepsilon^{\prime}} \\
    &\cdot \frac{w_1 N\cdot T\left(2+\frac{2}{3} \varepsilon^{\prime}\right)\left(\ln T^{k}+\ln \binom{N}{k}+l \ln N+\ln 2+\ln \log_{2} N\right)}{\varepsilon^{\prime 2} x_{i}}\\
    &\cdot \frac{1}{w_1 N\cdot T}\Bigg)\\
    &=\frac{1}{2 T^{k} \cdot\binom{N}{k} N^{l} \log_{2}N}
\end{aligned}
\end{displaymath}
For any $k$ size set $\cS$.
\begin{displaymath}
\begin{aligned}
    &\underset{\cR_{0} \sim \Omega}{\operatorname{Pr}}\left\{\hat{\rho}_{T\theta_{i}}\left(\cS, \cR_{0}\right) \geq\left(1+\varepsilon^{\prime}\right) \cdot\OPT\right\}\\
    &=\underset{\cR_{0} \sim \Omega}{\operatorname{Pr}}\left\{\frac{N\cdot T}{\theta_{i}} \cdot \sum_{j=1}^{\theta_{i}} Y_{j}^{\cR_{0}\left[\theta_{i}\right]}(\cS) \geq\left(1+\varepsilon^{\prime}\right) \cdot \OPT\right\}\\
    &\underset{\cR_{0} \sim \Omega}{\operatorname{Pr}}\left\{\sum_{j=1}^{\theta_{i}} Y_{j}^{\cR_{0}\left[\theta_{i}\right]}(\cS)-\frac{\theta_{i} \cdot \rho_T(\cS)}{N\cdot T} \geq \frac{\theta_{i}}{N\cdot T}\left(1+\varepsilon^{\prime}\right) \cdot \OPT-\frac{\theta_{i} \cdot \rho_T(\cS)}{N\cdot T}\right\}\\
    &/ *\OPT\geq \rho_T(\cS) * /\\
    &\leq \underset{\cR_{0} \sim \Omega}{\operatorname{Pr}}\left\{\sum_{j=1}^{\theta_{i}} Y_{j}^{\cR_{0}\left[\theta_{i}\right]}(\cS)-\frac{\theta_{i} \cdot \rho_T(\cS)}{N\cdot T} \geq \frac{\theta_{i}}{N\cdot T} \varepsilon^{\prime} \cdot \OPT\right\} \quad \\
\end{aligned}
\end{displaymath}
\begin{displaymath}
\begin{aligned}
    &\leq \underset{\cR_{0} \sim \Omega}{\operatorname{Pr}}\left\{\sum_{j=1}^{\theta_{i}} Y_{j}^{\cR_{0}\left[\theta_{i}\right]}(\cS)-\frac{\theta_{i} \cdot \rho_T(\cS)}{N\cdot T} \geq \frac{\varepsilon^{\prime} \cdot \OPT}{\rho_T(\cS)} \cdot \frac{\theta_{i} \cdot \rho_T(\cS)}{N\cdot T}\right\}\\
    &\leq \underset{\cR_{0} \sim \Omega}{\operatorname{Pr}}\left\{\sum_{j=1}^{\theta_{i}} \frac{Y_{j}^{\cR_{0}\left[\theta_{i}\right]}(\cS)}{w_1}-\frac{\theta_{i} \cdot \rho_T(\cS)}{w_1 N\cdot T} \geq \frac{\varepsilon^{\prime} \cdot \OPT}{\rho_T(\cS)} \cdot \frac{\theta_{i} \cdot \rho_T(\cS)}{w_1 N\cdot T}\right\}\\
    &\leq \exp \left(-\frac{\left(\frac{\varepsilon^{\prime} \cdot \OPT}{\rho_T(\cS)}\right)^{2}}{2+\frac{2}{3}\left(\frac{\varepsilon^{\prime} \cdot \OPT}{\rho_T(\cS)}\right)} \cdot \frac{\theta_{i} \cdot \rho_T(\cS)}{w_1 N\cdot T}\right) \quad \\
    &=\exp \left(-\frac{\left(\varepsilon^{\prime} \cdot \OPT\right)^{2}}{2 \rho_T(\cS)+\frac{2}{3}\left(\varepsilon^{\prime} \cdot \OPT\right)} \cdot \frac{\theta_{i}}{w_1 N\cdot T}\right)\\
    &\leq \exp \left(-\frac{\varepsilon^{\prime 2} \cdot \OPT}{2+\frac{2}{3} \varepsilon^{\prime}} \cdot \frac{\theta_{i}}{w_1 N\cdot T}\right)\\
    &\leq \exp \left(-\frac{\varepsilon^{\prime 2} \cdot x_{i}}{2+\frac{2}{3} \varepsilon^{\prime}} \cdot \frac{\theta_{i}}{w_1 N\cdot T}\right)\\
    &\leq \exp \Bigg(-\frac{\varepsilon^{\prime 2} \cdot x_{i}}{2+\frac{2}{3} \varepsilon^{\prime}}\\
    &\cdot \frac{w_1 N\cdot T\left(2+\frac{2}{3} \varepsilon^{\prime}\right)\left(\ln T^{k}+\ln \binom{N}{k}+l \ln N+\ln 2+\ln \log_{2} N\right)}{\varepsilon^{\prime 2} x_{i}}\\
    &\cdot \frac{1}{w_1n T}\Bigg)\\
    &=\frac{1}{2 T^{k} \cdot\binom{N}{k} N^{l} \log_{2}N}
\end{aligned}
\end{displaymath}

With the union bound $\underset{\cR_{0} \sim \Omega}{\operatorname{Pr}}\left\{\hat{\rho}_{T\theta_{i}}\left(\cS_{i}, \cR_{0}\right) \geq\left(1+\varepsilon^{\prime}\right) \cdot \OPT\right\} \leq \frac{1}{2 N^{l} \log_{2}N}$ \newline
\end{proof}
With Theorem~\ref{theorem:LB<OPT} we know that $LB$ is a lower bound of $\OPT$ with high probability, so $\theta$ satisfy the Corollary~\ref{Corollary:1}. Further we can know that the probability of $LB< \OPT$ is at least $1-1/2N^l$.

\begin{proof}
Set $LB_{i}=\frac{\hat{\rho}_{T\theta_{i}}\left(\cS_{i}, \cR_{0}\right)}{\left(1+\varepsilon^{\prime}\right)}$. 
When $\operatorname{OPT} \geq x_{\left\lfloor\log_{2}N\right\rfloor-1}$.Set $i\geq 1$ is the smallest index to make $\OPT\geq x_i$.

For any $i^{'} \leq i-1,\OPT < x_{i^{'}},$ 

For any $i^{''} > i-1,\OPT \geq x_{i^{''}},$

We define the event $\varepsilon$ as: for any $i^{'} \leq i-1$, $\hat{\rho}_{T\theta_{i^{\prime}}}\left(\cS_{i^{\prime}}, \cR_{0}\right)<\left(1+\varepsilon^{\prime}\right) x_{i^{\prime}}$, and for any $i^{''} \geq i$, $\hat{\rho}_{T\theta_{i^{\prime}}}\left(\cS_{i^{\prime}}, \cR_{0}\right) \geq \left(1+\varepsilon^{\prime}\right) x_{i^{\prime}}$ 

Notice that $i^{'},i^{''}\geq1$, so when $i=1$, $i^{'}$ is not exist. The part of Event $\varepsilon$ about $i^{'}$ is true. Event $\varepsilon$ is the event that we expected. Because as Event $\varepsilon$ happens, $LB=LB_{i}$ or $LB=1$.
So Event $\varepsilon$ indicate that $LB \leq\OPT$ so the upper bound of Event $\varepsilon$ not happen is that:
\begin{displaymath}
\begin{aligned}
    \underset{\cR_{0} \sim \Omega}{\operatorname{Pr}}\left\{\neg\varepsilon\right\} \leq &\sum_{i^{'}=1}^{i-1} \underset{\cR_{0} \sim \Omega}{\operatorname{Pr}}\left\{\hat{\rho}_{T\theta_{i^{'}}}\left(\cS_{i^{'}},\cR_0\right) \geq (1+\varepsilon^{'})x_{i^{'}}\right\} + \\
    &\sum_{i^{''}=i-1}^{\left \lfloor \log_{2}N \right \rfloor-1} \underset{\cR_{0} \sim \Omega}{\operatorname{Pr}}\left\{\hat{\rho}_{T\theta_{i^{''}}}(\cS_{i^{''}},\cR_0) \geq (1+\varepsilon^{'})\OPT\right\}
\end{aligned}
\end{displaymath}
With the above, we know that $\underset{\cR_{0} \sim \Omega}{\operatorname{Pr}}\left\{\hat{\rho}_{T\theta_{i_{'}}}(\cS_{i^{'}},\cR_0) \geq (1+\varepsilon^{'})x_{i^{'}}\right\} \leq \frac{1}{2N^{l}\log_2{N}}$. And $\underset{\cR_{0} \sim \Omega}{\operatorname{Pr}}\left\{\hat{\rho}_{T\theta_{i^{''}}}(\cS_{i^{''}},\cR_0) \geq (1+\varepsilon^{'})\OPT\right\} \leq \frac{1}{2N^{l}\log_2{N}}$. So $\underset{\cR_{0} \sim \Omega}{\operatorname{Pr}}\left\{\neg\varepsilon\right\} \leq \frac{1}{2N^{l}}$.

When $\OPT<x_{\left\lfloor \log_2n\right\rfloor-1}$. $LB=1$ with probability at least $1-\frac{1}{2N^l}$. 
\end{proof}

For any $\varepsilon>0,l>0$, PRM-IMM Guarantees that $\hat{\cS}^g$ is the $\frac{1}{2}-\varepsilon$ approximation of $\OPT$ with probability at least $1-\frac{1}{N^l}$.
Define the Event $\varepsilon$ as the $LB \leq\OPT$, and put the $\cR_{0}^{'}[\tilde{\theta}]$ to PRM-NodeSelection to get the seed set $\hat{\cS}^g$ is the $\frac{1}{2}-\varepsilon$ approximation of the PRM problem.
\begin{displaymath}
    \rho_T\left(\hat{\cS}^{g}\left(\cR_{0}^{\prime}[\tilde{\theta}]\right)\right) \geq\left(\frac{1}{2}-\varepsilon\right) \cdot \OPT
\end{displaymath}
with union bound:
\begin{displaymath}
\begin{aligned}
    &\underset{\cR_{0} \sim \Omega \space \cR_{0}^{'} \sim \Omega}{\operatorname{Pr}} \left\{\neg \varepsilon\right\} \leq \underset{\cR_{0} \sim \Omega \space \cR_{0}^{'} \sim \Omega}{\operatorname{Pr}}\left\{\operatorname{LB}> \operatorname{OPT}\right\}\\
    &+\underset{\cR_{0} \sim \Omega \space \cR_{0}^{'} \sim \Omega}{\operatorname{Pr}}\left\{\operatorname{LB} \leq \operatorname{OPT} \wedge \rho_T\left(\hat{\cS}^{g}\left(\cR_{0}^{\prime}[\tilde{\theta}]\right)\right)<\left(\frac{1}{2}-\varepsilon\right) \cdot \OPT\right\}
\end{aligned}
\end{displaymath}
And $\underset{\cR_{0} \sim \Omega \space \cR_{0}^{'} \sim \Omega}{\operatorname{Pr}} \left\{\operatorname{LB}>\operatorname{OPT}\right\}=\underset{\cR_{0} \sim \Omega}{\operatorname{Pr}}\{\operatorname{LB}>\operatorname{OPT}\} \leq \frac{1}{2 N^{l}}$.

Now we know that:

$\underset{\cR_{0} \sim \Omega \space \cR_{0}^{'} \sim \Omega}{\operatorname{Pr}}\left\{\operatorname{LB} \leq \operatorname{OPT} \wedge \rho_T\left(\hat{\cS}^{g}\left(\cR_{0}^{\prime}[\tilde{\theta}]\right)\right)<\left(\frac{1}{2}-\varepsilon\right) \cdot \OPT\right\}$

When $LB \leq\OPT$,$\tilde{\theta} \geq \frac{2 N \cdot\left[\frac{1}{2} \cdot \alpha+\beta\right]^{2}}{\varepsilon^{2} \cdot \OPT}$, and $\tilde{\theta} \leq \frac{2 N \cdot\left[\frac{1}{2} \cdot \alpha+\beta\right]^{2}}{\varepsilon^{2}}$.

$\theta_{\min }=\left\lceil\frac{2 N \cdot\left[\frac{1}{2} \cdot \alpha+\beta\right]^{2}}{\varepsilon^{2} \cdot \OPT} \right\rceil, \theta_{\max }=\left\lfloor\frac{2 N \cdot\left[\frac{1}{2} \cdot \alpha+\beta\right]^{2}}{\varepsilon^{2}} \right\rfloor$.$\tilde{\theta}$ is a integer range from $\theta_{min}$ to $\theta_{max}$.
\begin{displaymath}
\begin{aligned}
    &\underset{\cR_{0} \sim \Omega \space \cR_{0}^{'} \sim \Omega}{\operatorname{Pr}} \left\{\operatorname{LB} \leq \operatorname{OPT} \wedge \rho_T\left(\hat{\cS}^{g}\left(\cR_{0}^{\prime}[\tilde{\theta}]\right)\right)<\left(\frac{1}{2}-\varepsilon\right) \cdot \OPT\right\}\\
        &\leq \underset{\cR_{0} \sim \Omega \space \cR_{0}^{'} \sim \Omega}{\operatorname{Pr}} \left\{ \tilde{\theta} \geq \theta_{\min} \wedge \tilde{\theta} \leq \theta_{\max } \wedge \rho_T\left(\hat{\cS}^{g}\left(\cR_{0}^{\prime}[\tilde{\theta}]\right)\right)<\left(\frac{1}{2}-\varepsilon\right) \cdot \OPT\right\}\\
\end{aligned}
\end{displaymath}
\begin{displaymath}
\begin{aligned}
    &=\sum_{\theta=\theta_{\min }}^{\theta_{\max }} \underset{\cR_{0} \sim \Omega \space \cR_{0}^{'} \sim \Omega}{\operatorname{Pr}}\left\{\tilde{\theta}=\theta \wedge \rho_T\left(\hat{\cS}^{g}\left(\cR_{0}^{\prime}[\tilde{\theta}]\right)\right)<\left(\frac{1}{2}-\varepsilon\right) \cdot \OPT\right\}\\
    &=\sum_{\theta=\theta_{\min }}^{\theta_{\max }} \underset{\cR_{0} \sim \Omega \space \cR_{0}^{'} \sim \Omega}{\operatorname{Pr}}\left\{\tilde{\theta}=\theta \wedge \rho_T\left(\hat{\cS}^{g}\left(\cR_{0}^{\prime}[\theta]\right)\right)<\left(\frac{1}{2}-\varepsilon\right) \cdot \OPT\right\}\\
    &/ * \text { because } \cR_{0} \text { is independent with } \cR_{0}^{\prime} \text { * } /\\
    &=\sum_{\theta=\theta_{\min }}^{\theta_{\max }} \underset{\cR_{0} \sim \Omega}{Pr}\{\tilde{\theta}=\theta\} \cdot \underset{\cR_{0}^{'} \sim \Omega}{\operatorname{Pr}}\left\{\rho_T\left(\hat{\cS}^{g}\left(\cR_{0}^{\prime}[\theta]\right)\right)<\left(\frac{1}{2}-\varepsilon\right) \cdot \OPT\right\}\\ 
    &\leq \sum_{\theta=\theta_{\min }}^{\theta_{\max }} \underset{\cR_{0} \sim \Omega}{\operatorname{Pr}} \left\{\tilde{\theta}=\theta\right\} \cdot \frac{1}{2 N^{l}}=\frac{1}{2N^l}
\end{aligned}
\end{displaymath}
So $\underset{\cR_0\sim\Omega,\cR_{0}^{'}\sim\Omega}{\operatorname{Pr}}\left\{\neg\varepsilon\right\}\leq\frac{1}{N^l}$, which means that with probability at least $1-\frac{1}{N^l}$, the output of PRM-IMM $\hat{\cS}^g$ is the $\frac{1}{2}-\varepsilon$ approximation of $\OPT$.

\section{Time complexity of PRM-IMM}
The time complexity of RPM-IMM is $O((k+l)(M+N)T\log (NT)/\varepsilon^2)$. 
\begin{proof}
We use the Martingale theorem in the IMM algorithm\cite{tang15} to estimate the time complexity of the PRM-IMM algorithm.
Then the time complexity of the IMM algorithm is $O(\mathbb{E}[\overset{\_}{\theta}+\tilde{\theta}]\cdot (\EPT+1))$, where $\mathbb{E}[\overset{\_}{\theta}+\tilde{\theta}]$ is the overall number of PW-RR sets needed to be generated.

\begin{displaymath}
\mathbb{E}[\bar{\theta}+\tilde{\theta}] \leq \frac{8\left(\lambda^{*}+\lambda^{\prime}\right) \cdot\left(1+\varepsilon^{\prime}\right)^{2}}{ \cdot O P T}+2,
\end{displaymath}
where
\begin{displaymath}
\begin{aligned}
    &\lambda^{*}=\frac{4 w_1 N T \cdot(\cdot \alpha+\beta)^{2}}{\varepsilon^{2}}\\ 
    &\lambda^{\prime}=\frac{w_1 N T \cdot\left(2+\frac{2}{3} \varepsilon^{\prime}\right) \cdot\left(\ln \binom{N}{k}+\ell \ln N+\ln 2+\ln \log_{2} N + \ln T^k\right)}{\varepsilon^{\prime 2}},
\end{aligned}
\end{displaymath}
and $\alpha$ and $\beta$ is defined in section 5.

Therefore, $\mathbb{E}[\bar{\theta}+\tilde{\theta}] = O\left(\frac{w_1(k+l)NT\log {NT}}{\OPT \varepsilon^2}\right)$.
And $\EPT = \frac{M}{N}\cdot\mathbb{E}[\sigma(v)]$ is the expected running time of generating a PW-RR set.
Because $\mathbb{E}[\sigma(v)] \leq \frac{\OPT}{w_1}$, so the expected running time is:
\begin{displaymath}
O\left(\frac{(k+l)(N+M)T\log {NT}}{\varepsilon^2}\right)
\end{displaymath}

\end{proof}

\section{Objective Function When Natural Growth Count is Variable(Proof of section~\ref{sec:variable z})}
We use the natural growth vector : $\boldsymbol{z}=\left[z_1, z_2, \cdots, z_t\right]$ with $z_t$ denoting the natural customer count in round $t$.
\begin{displaymath}
\begin{aligned}
& d_1^n=d_0^n+z_1 \cdot \frac{d_0^n}{d_0^n+d_0^p}+\sigma\left(S_1\right) \\
& d_1^p=d_0^p+z_1 \cdot \frac{d_0^p}{d_0^n+d_0^p} \\
& r_1=\frac{d_0^n+z_1 \cdot \frac{d_0^n}{d_0^n+d_0^p}+\sigma\left(S_1\right)}{d_0^p+z_1 \cdot \frac{d_0^p}{d_0^n+d_0^p}} \\
& r_1=r_0+\frac{\left(r_0+1\right) \sigma\left(S_1\right)}{d_0^n+d_0^p+z_1} \\
\end{aligned}
\end{displaymath}
\begin{displaymath}
\begin{aligned}
& r_t=r_{t-1}+\frac{\sigma\left(S_{t-1}\right)}{d_{t-1}^n+d_{t-1}^p+z_t}\left(r_{t-1}+1\right) \\
& r_t+1=r_{t-1}+\frac{\sigma\left(S_{t-1}\right)}{d_{t-1}^n+d_{t-1}^p+z_t}\left(r_{t-1}+1\right)+1 \\
& r_t+1=\left(1+\frac{\sigma\left(S_{t-1}\right)}{d_{t-1}^n+d_{t-1}^p+z_t}\right)\left(r_{t-1}+1\right) \\
& r_t+1=\left(1+\frac{\sigma\left(S_{t-1}\right)}{d_{t-1}^n+d_{t-1}^p+z_t}\right)\left(r_{t-1}+1\right) \\
& r_T(S)=\left(r_0+1\right) \prod_{t=1}^T\left(1+\frac{\sigma\left(S_t\right)}{d_0^n+d_0^p+\sum_{i=1}^t z_i+\sum_{i=1}^{t-1} \sigma\left(S_i\right)}\right)-1 \\
&
\end{aligned}
\end{displaymath}

\section{Upper and lower bound of Objective Function (Popular Item promotion section~\ref{sec:popularPromotion})}
The objective function $r_T(\cS)$ is not easily derived. So we can obtain its upper and lower bound as follows.

Upper bound:
\begin{displaymath}
\begin{aligned}
r_t+1&=\frac{d_t^n+d_t^p}{d_t^p}=\frac{d_{t-1}^n+d_{t-1}^p+\sigma\left(S_t\right)+z+p_t}{d_{t-1}^p+z \cdot \frac{d_{t-1}^p}{d_{t-1}^n+d_{t-1}^p}+p_t} \\
&=\frac{d_{t-1}^n+d_{t-1}^p+\sigma\left(S_t\right)+z+p_t}{d_{t-1}^n+d_{t-1}^p+z+p_t \cdot \frac{d_{t-1}^n+d_{t-1}^p}{d_{t-1}^p} \cdot \frac{d_{t-1}^n+d_{t-1}^p}{d_{t-1}^p}} \\
&<\frac{d_{t-1}^n+d_{t-1}^p+\sigma\left(S_t\right)+z+p_t}{d_{t-1}^n+d_{t-1}^p+z+p_t} \cdot \frac{d_{t-1}^n+d_{t-1}^p}{d_{t-1}^p} \\
&=\left(1+\frac{\sigma\left(S_t\right)}{d_{t-1}^n+d_{t-1}^p+z+p_t}\right) \cdot\left(r_{t-1}+1\right) \\
&r_t+1<\left(1+\frac{\sigma\left(S_t\right)}{d_0^n+d_0^p+z \cdot t+\sum_1^{t-1} \sigma\left(S_i\right)+\sum_1^t p_i}\right) \cdot\left(r_{t-1}+1\right) \\
\end{aligned}
\end{displaymath}

\begin{displaymath}
\begin{aligned}
&r_{t-1}+1<\left(1+\frac{\sigma\left(S_{t-1}\right)}{d_0^n+d_0^p+z \cdot(t-1)+\sum_1^{t-2} \sigma\left(S_i\right)+\sum_1^{t-1} p_i}\right)\\
& \cdot\left(r_{t-2}+1\right) \\
&r_T+1<\left(r_0+1\right) \prod_{t=1}^T\left(1+\frac{\sigma\left(S_t\right)}{d_0^n+d_0^p+z \cdot t+\sum_1^{t-1} \sigma\left(S_i\right)+\sum_1^t p_i}\right) \\
&r_T<r_T^{\prime}\\
&=\left(r_0+1\right) \prod_{t=1}^T\left(1+\frac{\sigma\left(S_t\right)}{d_0^n+d_0^p+z \cdot t+\sum_1^{t-1} \sigma\left(S_i\right)+\sum_1^t p_i}\right)-1
\end{aligned}
\end{displaymath}

Lower bound:
\begin{displaymath}
\begin{aligned}
& r_t+1=\frac{d_t^n+d_t^p}{d_t^p}=\frac{d_{t-1}^n+d_{t-1}^p+\sigma\left(S_t\right)+z+p_t}{d_{t-1}^p+z \cdot \frac{d_{t-1}^p}{d_{t-1}^n+d_{t-1}^p}+p_t} \\
& =\frac{d_{t-1}^n+d_{t-1}^p+\sigma\left(S_t\right)+z+p_t}{d_{t-1}^n+d_{t-1}^p+z+p_t \cdot \frac{d_{t-1}^n+d_{t-1}^p}{d_{t-1}^p}} \cdot \frac{d_{t-1}^n+d_{t-1}^p}{d_{t-1}^p} \\
& >\frac{d_{t-1}^n+d_{t-1}^p+\sigma\left(S_t\right)+z+p_t}{d_{t-1}^n+d_{t-1}^p+z+2 \cdot p_t} \cdot \frac{d_{t-1}^n+d_{t-1}^p}{d_{t-1}^p} \\
& =\left(1+\frac{\sigma\left(S_t\right)-p_t}{d_{t-1}^n+d_{t-1}^p+z+2 \cdot p_t}\right) \cdot\left(r_{t-1}+1\right) \\
& r_t+1>\left(1+\frac{\sigma\left(S_t\right)-p_t}{d_0^n+d_0^p+z \cdot t+\sum_1^{t-1} \sigma\left(S_i\right)+\sum_1^t p_i+p_t}\right)\\
&\cdot\left(r_{t-1}+1\right) \\
& r_{t-1}+1>\left(1+\frac{\sigma\left(S_{t-1}\right)-p_{t-1}}{d_0^n+d_0^p+z \cdot(t-1)+\sum_1^{t-2} \sigma\left(S_i\right)+\sum_1^{t-1} p_i+p_{t-1}}\right)\\
&\cdot\left(r_{t-2}+1\right) \\
& r_T>\left(r_0+1\right) \prod_{t=1}^T\left(1+\frac{\sigma\left(S_t\right)-p_t}{d_0^n+d_0^p+z \cdot t+\sum_1^{t-1} \sigma\left(S_i\right)+\sum_1^t p_i+p_t}\right)-1 \\
& r_T>r_T^{\prime \prime}\\
&=\left(r_0+1\right) \prod_{t=1}^T\left(1+\frac{\sigma\left(S_t\right)-p_t}{d_0^n+d_0^p+z \cdot t+\sum_1^{t-1} \sigma\left(S_i\right)+\sum_1^t p_i+p_t}\right)-1 \\
&
\end{aligned}
\end{displaymath}

\section{Upper and lower bound of Objective Function (Multi-Item promotion Theorem~\ref{theorem:multi-item promotion})}
In the setting of multiple items with promotions, we can derive the bound of $r_T$ as follows.

Upper bound:
\begin{displaymath}
\begin{aligned}
& r_t+1=\frac{d_t^n+\sum_{j=1}^{s}d_{t}^{p^j}}{\sum_{j=1}^{s}d_{t}^{p^j}}=\frac{d_{t-1}^n+\sum_{j=1}^{s}d_{t-1}^{p^j}+\sigma\left(S_t\right)+z+\sum_{j=1}^{s}p^j}{\sum_{j=1}^{s}d_{t-1}^{p^j}+z \cdot \frac{\sum_{j=1}^{s}d_{t-1}^{p^j}}{d_{t-1}^n+\sum_{j=1}^{s}d_{t-1}^{p^j}}+\sum_{j=1}^{s}p^j} \\
\end{aligned}
\end{displaymath}

\begin{displaymath}
\begin{aligned}
& =\frac{d_{t-1}^n+\sum_{j=1}^{s}d_{t-1}^{p^j}+\sigma\left(S_t\right)+z+\sum_{j=1}^{s}p^j}{d_{t-1}^n+\sum_{j=1}^{s}d_{t-1}^{p^j}+z+\left(\sum_{j=1}^{s}p^j\right) \cdot \frac{d_{t-1}^n+\sum_{j=1}^{s}d_{t-1}^{p^j}}{\sum_{j=1}^{s}d_{t-1}^{p^j}}} \cdot \frac{d_{t-1}^n+\sum_{j=1}^{s}d_{t-1}^{p^j}}{\sum_{j=1}^{s}d_{t-1}^{p^j}} \\
& <\frac{d_{t-1}^n+\sum_{j=1}^{s}d_{t-1}^{p^j}+\sigma\left(S_t\right)+z+\sum_{j=1}^{s}p^j}{d_{t-1}^n+\sum_{j=1}^{s}d_{t-1}^{p^j}+z+\left(\sum_{j=1}^{s}p^j\right)} \cdot \frac{d_{t-1}^n+\sum_{j=1}^{s}d_{t-1}^{p^j}}{\sum_{j=1}^{s}d_{t-1}^{p^j}} \\
& =\left(1+\frac{\sigma\left(S_t\right)}{d_{t-1}^n+\sum_{j=1}^{s}d_{t-1}^{p^j}+z+\sum_{j=1}^{s}p^j}\right) \cdot\left(r_{t-1}+1\right) \\
\end{aligned}
\end{displaymath}

\begin{displaymath}
\begin{aligned}
&r_t+1\\
&<\left(1+\frac{\sigma\left(S_t\right)}{d_0^n+\sum_{j=1}^{s}p^j+z \cdot t+\sum_1^{t-1} \sigma\left(S_i\right)+\sum_1^t p_i}\right) \cdot\left(r_{t-1}+1\right) \\
&r_{t-1}+1<\\
&\left(1+\frac{\sigma\left(S_{t-1}\right)}{d_0^n+\sum_{j=1}^{s}p^j+z \cdot(t-1)+\sum_1^{t-2} \sigma\left(S_i\right)+\sum_1^{t-1} p_i}\right) \cdot\left(r_{t-2}+1\right) \\
&r_T+1\\
&<\left(r_0+1\right) \prod_{t=1}^T\left(1+\frac{\sigma\left(S_t\right)}{d_0^n+\sum_{j=1}^{s}p^j+z \cdot t+\sum_1^{t-1} \sigma\left(S_i\right)+\sum_1^t p_i}\right) \\
&r_T<r_T^{\prime}\\
&=\left(r_0+1\right) \prod_{t=1}^T\left(1+\frac{\sigma\left(S_t\right)}{d_0^n+\sum_{j=1}^{s}p^j+z \cdot t+\sum_1^{t-1} \sigma\left(S_i\right)+\sum_1^t p_i}\right)-1
\end{aligned}
\end{displaymath}

Lower bound:
\begin{displaymath}
\begin{aligned}
& r_t+1=\frac{d_t^n+d_t^p}{d_t^p}=\frac{d_{t-1}^n+d_{t-1}^p+\sigma\left(S_t\right)+z+p_t}{d_{t-1}^p+z \cdot \frac{d_{t-1}^p}{d_{t-1}^n+d_{t-1}^p}+\sum_{j=1}^{s}d_{t}^{p^j}} \\
& =\frac{d_{t-1}^n+d_{t-1}^p+\sigma\left(S_t\right)+z+\sum_{j=1}^{s}d_{t}^{p^j}}{d_{t-1}^n+d_{t-1}^p+z+\sum_{j=1}^{s}d_{t}^{p^j} \cdot \frac{d_{t-1}^n+d_{t-1}^p}{d_{t-1}^p}} \cdot \frac{d_{t-1}^n+d_{t-1}^p}{d_{t-1}^p} \\
& >\frac{d_{t-1}^n+d_{t-1}^p+\sigma\left(S_t\right)+z+\sum_{j=1}^{s}d_{t}^{p^j}}{d_{t-1}^n+d_{t-1}^p+z+2 \cdot \sum_{j=1}^{s}d_{t}^{p^j}} \cdot \frac{d_{t-1}^n+d_{t-1}^p}{d_{t-1}^p} \\
& =\left(1+\frac{\sigma\left(S_t\right)-\sum_{j=1}^{s}d_{t}^{p^j}}{d_{t-1}^n+d_{t-1}^p+z+2 \cdot \sum_{j=1}^{s}d_{t}^{p^j}}\right) \cdot\left(r_{t-1}+1\right) \\
& r_t+1>\left(1+\frac{\sigma\left(S_t\right)-\sum_{j=1}^{s}d_{t}^{p^j}}{d_0^n+\sum_{j=1}^{s}p^j+z \cdot t+\sum_1^{t-1} \sigma\left(S_i\right)+\sum_1^t p_i+\sum_{j=1}^{s}d_{t}^{p^j}}\right)\\
&\cdot\left(r_{t-1}+1\right) \\
\end{aligned}
\end{displaymath}
\begin{displaymath}
\begin{aligned}
& r_{t-1}+1>\left(1+\frac{\sigma\left(S_{t-1}\right)-p_{t-1}}{d_0^n+\sum_{j=1}^{s}p^j+z \cdot(t-1)+\sum_1^{t-2} \sigma\left(S_i\right)+\sum_1^{t-1} p_i+p_{t-1}}\right)\\
&\cdot\left(r_{t-2}+1\right) \\
& r_T>\left(r_0+1\right) \\
&\prod_{t=1}^T\left(1+\frac{\sigma\left(S_t\right)-\sum_{j=1}^{s}d_{t}^{p^j}}{d_0^n+\sum_{j=1}^{s}p^j+z \cdot t+\sum_1^{t-1} \sigma\left(S_i\right)+\sum_1^t p_i+\sum_{j=1}^{s}d_{t}^{p^j}}\right)-1 \\
& r_T>r_T^{\prime \prime}=\left(r_0+1\right)\\
&
\prod_{t=1}^T\left(1+\frac{\sigma\left(S_t\right)-\sum_{j=1}^{s}d_{t}^{p^j}}{d_0^n+\sum_{j=1}^{s}p^j+z \cdot t+\sum_1^{t-1} \sigma\left(S_i\right)+\sum_1^t p_i+\sum_{j=1}^{s}d_{t}^{p^j}}\right)-1 \\
\end{aligned}
\end{displaymath}

\end{document}